\documentclass{article}
\usepackage[letterpaper, margin=1in]{geometry}

\usepackage[T1]{fontenc}
\usepackage{lmodern}

\usepackage{adjustbox}
\usepackage{comment}
\usepackage{soul}
\usepackage{graphicx}
\usepackage{amsmath}
\usepackage{amsfonts}
\usepackage{amsthm}
\usepackage{amssymb}
\usepackage[inline]{enumitem}
\usepackage{mathtools}
\usepackage{xcolor}
\usepackage{multirow}

\usepackage{titlesec}
\titlespacing*{\paragraph}
{0pt}{1.5ex plus 1ex minus .2ex}{1em}

\usepackage{subcaption}
\usepackage{bbold}
\usepackage{bm}

\usepackage{tikz}
\usetikzlibrary{arrows}
\usetikzlibrary{backgrounds}
\usetikzlibrary{calc}
\usetikzlibrary{decorations.pathreplacing}
\usetikzlibrary{intersections}
\usetikzlibrary{fit}
\usetikzlibrary{patterns}
\usetikzlibrary{positioning}
\usetikzlibrary{shapes}
\usetikzlibrary{shadows}
\usetikzlibrary{shapes.geometric}
\usetikzlibrary{tikzmark}

\usepackage{hyperref}
\hypersetup{
    colorlinks,
    linkcolor={blue},
    citecolor={black},
    urlcolor={blue}}

\usepackage{doi}
\usepackage[backend=biber,style=alphabetic,maxbibnames=6,natbib]{biblatex}

\usepackage{zref-clever}
\NewDocumentCommand{\cref}{m}{\zcref{#1}}
\NewDocumentCommand{\Cref}{m}{\zcref[cap=true]{#1}}

\zcsetup{nameinlink=false}
\zcLanguageSetup{english}{cap}
\zcRefTypeSetup{equation}{
    Name-sg=Equation,
    name-sg=,
    Name-pl=Equations,
    name-pl=,
    Name-sg-ab=Eq.,
    name-sg-ab=,
    Name-pl-ab=Eqs.,
    name-pl-ab=,
    abbrev=true,
    nocap,
}

\newcommand*{\defeq}{\coloneq}

\usepackage{algorithm}
\usepackage[indLines=true,noEnd=true,commentColor=darkgray,endLComment={}]{algpseudocodex}
\algnewcommand{\IfThenElse}[3]{\algorithmicif\ #1\ \algorithmicthen\ #2\ \algorithmicelse\ #3}
\makeatletter\AddToHook{env/algorithmic/begin}{\zcsetup{currentcounter=ALG@line}}\makeatother
\zcsetup{countertype={ALG@line=line}}
\zcRefTypeSetup{line}{nocap}

\usepackage{keytheorems}
\newkeytheorem{theorem}[style=plain,numberwithin=section]
\newkeytheorem{lemma}[style=plain,sibling=theorem]
\newkeytheorem{proposition}[style=plain,sibling=theorem]
\newkeytheorem{corollary}[style=plain,sibling=theorem]
\newkeytheorem{definition}[style=definition,sibling=theorem]
\newkeytheorem{remark}[style=definition,sibling=theorem,numbered=false]
\newkeytheorem{notation}[style=definition,sibling=theorem]
\newkeytheorem{observation}[style=definition,sibling=theorem]

\makeatletter
\newcommand{\zcBookmarkNumber}[1]{%
  \zref@extractdefault{#1}{default}{??}%
}
\newcommand{\zcBookmarkTypeName}[1]{%
  \@ifundefined{zcBookmarkType@\zref@extractdefault{#1}{zc@type}{}}%
    {\zref@extractdefault{#1}{zc@type}{??}}%
    {\@nameuse{zcBookmarkType@\zref@extractdefault{#1}{zc@type}{}}}%
}
\newcommand{\zcBookmarkType@proposition}{Proposition}
\newcommand{\zcBookmarkType@theorem}{Theorem}
\newcommand{\zcBookmarkType@lemma}{Lemma}
\newcommand{\zcBookmarkType@corollary}{Corollary}
\newcommand{\zcBookmarkType@definition}{Definition}
\newcommand{\zcBookmarkType@notation}{Notation}
\newcommand{\zcBookmarkType@observation}{Observation}
\newcommand{\titlezcref}[1]{%
  \texorpdfstring
    {\zcref[S]{#1}}%
    {\zcBookmarkTypeName{#1}~\zcBookmarkNumber{#1}}%
}
\makeatother

\DeclareMathOperator{\poly}{poly}

\NewDocumentCommand{\F}{sO{F}m m}{\IfBooleanTF{#1}{\widetilde{#2}}{#2}\!\left(#3,#4\right)}
\NewDocumentCommand{\Fl}{O{F}m m}{\F[#1^{-}]{#2}{#3}}
\NewDocumentCommand{\Fu}{O{F}m m}{\F[#1^{+}]{#2}{#3}}
\NewDocumentCommand{\Fp}{O{F}m m}{\F[#1^{\pm}]{#2}{#3}}
\NewDocumentCommand{\CF}{O{F}m m}{{\mathcal C}_{#1}\!\left(#2,#3\right)}

\usepackage{etoolbox}
  \makeatletter
    \patchcmd{\maketitle}{\@fnsymbol}{\@arabic}{}{}
  \makeatother

\DeclarePairedDelimiter{\set}{\lbrace}{\rbrace}
\DeclarePairedDelimiter{\ceil}{\lceil}{\rceil}
\DeclarePairedDelimiter{\floor}{\lfloor}{\rfloor}
\DeclarePairedDelimiter{\abs}{\lvert}{\rvert}

\newcommand{\Nat}{\mathbb{N}}

\newcommand{\Real}{\mathbb{R}}
\newcommand{\Rat}{\mathbb{Q}}

\newcommand{\Uniform}{\mathrm{Unif}}

\DeclarePairedDelimiter{\brak}{[}{]}
\NewDocumentCommand{\expect}{O{}m}{\mathbb{E}\IfBlankF{#1}{_{#1}}\brak*{#2}}

\usepackage[most]{tcolorbox}
\makeatletter
\newtcolorbox{nicebox}[1][]{%
  colback=white,          
  colframe=black,         
  after app={\@endparenv},
  #1
}
\makeatother

\newcommand{\varlo}{X^-}
\newcommand{\varhi}{X^+}
\newcommand{\varmid}{X^{\text{mid}}}

\title{Online Random Sampling with Real Probabilities}

\begin{document}

\author{%
  Thomas L.~Draper\thanks{Carnegie Mellon University. Email: \textsf{\{\href{mailto:tdraper@cmu.edu}{tdraper},\href{mailto:fsaad@cmu.edu}{fsaad}\}@cmu.edu}} \and
  David G.~Harris\thanks{University of Maryland. Email: \textsf{\href{mailto:davidgharris29@gmail.com}{davidgharris29}@gmail.com}} \and
  Feras A.~Saad\footnotemark[1]}

\maketitle

\vspace{-.5cm}

\begin{abstract}
We develop an efficient online algorithm to sample a
sequence of discrete random variables
using an entropy source of i.i.d.~fair coin flips,
in a standard model of real computation where real-valued
probabilities are represented by rational approximations.
For any sequence $F_1, F_2, \dots$ of probability distributions, our algorithm
generates $n$ outputs $X_1 \sim F_1, \dots, X_n \sim F_n$ using at
most $\expect{H(F_1) +\dots + H(F_n)} + O(\log n)$ coin flips in
expectation while carrying $O(\log n)$ bits of persistent space,
where $H$ is the Shannon entropy.
Under standard assumptions, we prove that our algorithm achieves
this information-theoretically optimal entropy rate using
asymptotically optimal space.

The key idea is to replace the global arithmetic-decoding sampling scheme
of \citeauthor{han1997} (\citeyear{han1997}) with a local discrete uniform
state, yielding an exponential reduction in space for a given entropy loss.
Our approach applies to distributions with irrational probabilities and
countably infinite supports, generalizing recent randomness-recycling
methods beyond finite rational distributions with bounded denominator.
\end{abstract}

\pagenumbering{arabic}

\section{Introduction}

A fundamental problem in theoretical and applied computer science is
to convert a sequence of random symbols drawn from an ``input'' stream into
a sequence of random symbols that form the ``output'' stream.
In \textit{random sampling}, the input stream is a sequence $B_1, B_2,
\dots$ of i.i.d.~fair coin flips and the output stream is a sequence
$X_1, X_2, \dots$ of random variables~\citep{knuth1976}.
This operation is used in almost all computer
simulations, such as algorithms for stochastic optimization or numerical
integration that demand many random variables drawn from
a dynamic sequence of target distributions.

The most flexible workflow for applications is \textit{online} random sampling,
where each target distribution
$F_i$ for $X_i$ is revealed at stage $i$ of a sampling loop,
as shown in the following high-level pseudocode.
\begin{nicebox}[frame empty,before skip=0pt]
\begin{algorithmic}
\State \textbf{Input:} Sequence $B_1, B_2, \dots$ of i.i.d.~coin flips
\While{true}
\State \quad \textbf{Receive} target distribution $F$
  \Comment{may depend on previous outputs or exogenous randomness}
\State \quad \textbf{Emit} random variable $X \sim F$
\EndWhile
\end{algorithmic}
\end{nicebox}
At iteration $i$ of the sampling loop, the algorithm must produce exact samples
$X_i \sim F_i$, even if conditioned on previous outputs $X_1,\ldots,X_{i-1}$.
The total number of requests may not be known in advance and may be unbounded.
The algorithm may maintain some private internal state from round-to-round,
and the rule for selecting each $F_i$ must be independent of this private
state conditioned on $X_1, \ldots, X_{i-1}$.

Any algorithm that converts input coin tosses to other output variables
typically incurs a loss of information, which is a direct consequence of
the data processing inequality.
More specifically, the Shannon entropy contained in the output variables is
at most that of the consumed input coin flips.
For example, consider rolling a fair six-sided die by
first flipping three fair coins and then applying the mapping
\begin{equation}
\set*{000 \mapsto 1, \quad
001 \mapsto 2, \quad
010 \mapsto 3, \quad
011 \mapsto 4, \quad
100 \mapsto 5, \quad
101 \mapsto 6, \quad
110 \mapsto \mbox{R}, \quad
111 \mapsto \mbox{R}},
\end{equation}
where 110 and 111 invoke a ``reject and repeat'' loop.
It takes 4/3 trials to halt on average and each trial
uses 3 bits of entropy, for a total of 4 input bits.
The output roll carries $\lg 6$ bits of entropy, so the expected entropy
loss is $4 - \lg 6 \approx 1.4150$ bits.
A more efficient algorithm is to recycle an unused bit when rejecting:
\begin{equation}
\set*{000 \mapsto 1, \quad
001 \mapsto 2, \quad
010 \mapsto 3, \quad
011 \mapsto 4, \quad
100 \mapsto 5, \quad
101 \mapsto 6, \quad
110 \mapsto \mbox{R0}, \quad
111 \mapsto \mbox{R1}}.
\end{equation}

After rejecting with $110$ or $111$, the final bit can be reused,
so that the next loop iteration only needs to draw \textit{two} fresh bits.
Now only $11/3$ input flips are used on average, which saves 1/3 bit.
As millions or even billions of outputs are generated, small differences in
entropy consumption can accumulate into large differences in runtime.
This overhead is most prominent in applications where the entropy source is
expensive, such as cryptographic protocols requiring high-quality random
bits~\citep{rfc4086}.
We thus view the entropy loss as a fundamental quantity to be minimized, along
with space and runtime.

\paragraph{Entropy Cost.}
Formally, the $n$-step entropy loss of an online random sampling algorithm
is the difference between the expected
Shannon entropy of the input coin flips and the output symbols:
\begin{align}
L_n \defeq \expect{ T_n - \sum_{i=1}^n H(F_i) } \ge 0
&& (n \ge 1),
\label{eq:entropy-loss}
\end{align}
where $T_n$ is the random number of coin flips (each carrying 1 bit of
information) used to produce the $n$ outputs and
$H(F) \defeq -\sum_{i} f(i) \lg f(i)$ is the Shannon entropy
of a discrete cumulative distribution function (CDF)
$F$ with probability mass function (PMF)
$f(i) \defeq F(i) - F(i-1)$.
We aim to minimize $\expect{T_n}$, and thus $L_n$, while retaining
efficient space and runtime.
%

\paragraph{Real-Valued Probabilities.}
Many random sampling algorithms assume the Real RAM model.
For example, the standard treatment in the statistical literature---where
random sampling is often called ``random variate generation''---assumes
that arithmetic operations on real numbers take $O(1)$
time~\citep[\S1.1; Assumptions I--III]{devroye1986}.
In this setting, it is trivial to achieve $L_n < 3$ via arithmetic decoding
of the input bit stream~\citep{han1997},
nearly matching the optimal entropy loss $L_n < 2$ given in \citet{knuth1976}.

We aim to develop algorithms that handle distributions with arbitrary
(irrational) probabilities in a more realistic model of computation than
Real RAM.
Following the standard model~\citep[\S2]{ko1991},
we access each target distribution $F$
via an approximation oracle $\phi$ for its CDF,
such that $\abs{F(x) - \phi(x,k)} < 2^{-k}$ for any
$x,k \in \Nat \defeq \set{0, 1, \dots }$.
This assumption captures the usual situation for many practical distributions,
where probabilities are not necessarily rational but can be computed to
any prescribed precision.

In our setting, efficient sampling becomes significantly more challenging than
``dice rolling'' (i.e., sampling a finite-support distribution with
rational probabilities), because there are now two sources of uncertainty
\begin{enumerate*}[label=(\roman*)]
\item the infinitely precise uniform random real $U \in [0,1]$, whose bits $B_1, B_2, \dots$
  are presented as the i.i.d.~coin flips from the entropy source; and
\item the cumulative probabilities $F(x)$ in the target distribution, which are
  real numbers known up to finite precision after a finite number of oracle queries.
\end{enumerate*}
An efficient sampling algorithm must carefully balance the cost of drawing
fresh coin flips from the entropy source with the cost of invoking the
oracle at higher and higher precision to control the space, time,
and entropy costs.

\paragraph{Space Cost.}
We distinguish \textit{persistent space}, which is the number of bits used
for the state persisting between the iterations of online sampling, and
\textit{temporary space}, which is used only within a single sampling
iteration.
The \textit{overall} space is the maximum of these two quantities across all
iterations of online sampling.
If the distributions have irrational probabilities, then the
temporary space must be unbounded to guarantee exact outputs,
although the \textit{expected} temporary space may remain bounded.
On the other hand, there is no restriction on the persistent space: it is
zero for algorithms that waste entropy by maintaining no state,
and it grows unbounded as $n \to \infty$ for classical
algorithms achieving (near) optimal entropy loss~\citep{knuth1976,han1997}.

\paragraph{Global State for Online Random Sampling.}
The dice-rolling example illustrates a general phenomenon: reducing entropy
loss requires preserving randomness that was not fully used by the sampler.
More generally, if a sampler for a CDF $F$ outputs $X=x$ using $U \sim \Uniform[0,1]$,
then the conditional position of $U$ inside the interval $[F(x-1),F(x)]$ remains
uniform.
This information constitutes ``unused'' randomness that can be used in
future sampling rounds to reduce the number of fresh bits drawn from the entropy source.

A key obstacle to efficient online sampling is that existing methods reduce
entropy cost by storing reusable randomness from the entire output history.
In online sampling with arithmetic decoding~\citep{han1997}, the
reusable randomness is maintained \textit{globally}, by carrying
a high-precision interval containing $U$ after producing $n$
outputs $X_1, \dots, X_n$.
For example, if all the probabilities are rationals with denominator $d$, then endpoints
have a denominator growing as $d^n$, and the
expected space is $O(n)$.

\paragraph{Key Idea: Local Randomness Recycling.}
Instead of storing a large global interval as in arithmetic coding,
we maintain a \textit{local} discrete state which is a lossy,
compressed form of the unused randomness from the previous samples.
Specifically, after each step $i$, our sampler's state is a pair of integers
$(Z_{i}, M_{i})$ with $Z_{i} \mid M_{i} \sim \Uniform\set{0,\dots,M_{i}-1}$.
We sample $X_{i+1}$ using a uniform variate $Q_{i+1}$ constructed
from $(Z_{i}, M_{i})$ and fresh coin flips from the source as needed.
Then, a new uniform state $(Z_{i+1}, M_{i+1})$ is obtained by compactly
encoding the \textit{relative} position of $Q_{i+1}$ within the interval
$[F_{i+1}(x_{i+1}-1),F_{i+1}(x_{i+1})]$.

The state is not obtained by simply truncating high-precision interval
endpoints, which introduces bias.
Instead, we use a discrete re-encoding that preserves exactness while
bounding the state size.
To control the space, we set an upper bound $M_i \leq \Delta_i$
that evolves with the iteration count $i$.
Choosing the sequence $\Delta_i$ is a delicate balance:
larger values give a lower asymptotic entropy loss but increase
the runtime.
Large values also create leftover randomness that is
discarded at the (unknown) end of the online sampling procedure.

In our algorithm, we choose $\Delta_i$ to scale as $i \log i$ up to
some optional specified maximum value, after which it stabilizes.
Over $n$ rounds, the entropy loss is $L_n \leq \epsilon n + O(\log n)$ and the
persistent space is $O(\log \min\set{1/\epsilon, n})$,
which is an exponential improvement over the global-state approach.
Here $\epsilon$ can be any desired user-specified parameter, including
$\epsilon = 0$ to obtain a purely logarithmic entropy loss.

\paragraph{Overview of the Algorithm.}

\Cref{fig:demo} gives an overview of the randomness-recycling procedure.


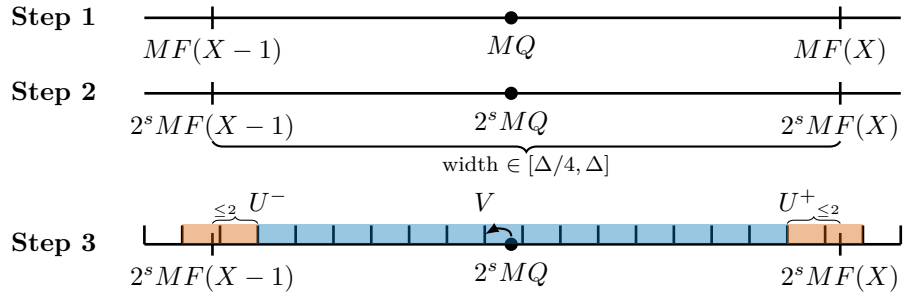
\begin{figure}[H]
\centering
\definecolor{cbBlue}{HTML}{0072B2}
\definecolor{cbRed}{HTML}{D55E00}
\begin{tikzpicture}[thick]
\tikzset{title/.style={anchor=west, fill=gray, fill opacity=.10, text opacity=1, minimum width=10cm,}}

\node[at={(-0.5,0)},anchor=east,font=\bfseries]{Step 1};
\draw[-, line width=1pt] (0,0) -- (10,0)
  node[name=fa,pos=0.09,anchor=center,label={below:$MF(X-1)$}]{}
  node[name=fb,pos=0.92,anchor=center,label={below:$MF(X)$}]{}
  node[name=fq,pos=0.485,anchor=center,circle,minimum size=5pt,fill=black,inner sep=0pt,label={below:$MQ$}]{}
  ;


\draw[thick, line width=1pt] ([yshift=-.15cm]fa.center) -- ([yshift=.15cm]fa.center);
\draw[thick, line width=1pt] ([yshift=-.15cm]fb.center) -- ([yshift=.15cm]fb.center);

\node[at={(-0.5,-1)},anchor=east,font=\bfseries]{Step 2};
\draw[-, line width=1pt] (0,-1) -- (10,-1)
  node[name=ga,pos=0.09,anchor=center,label={below:$2^sMF(X-1)$}]{}
  node[name=gb,pos=0.92,anchor=center,label={below:$2^sMF(X)$}]{}
  node[name=gq,pos=0.485,anchor=center,circle,minimum size=5pt,fill=black,inner sep=0pt,label={below:$2^sMQ$}]{}
  ;
\draw[thick, line width=1pt] ([yshift=-.15cm]ga.center) -- ([yshift=.15cm]ga.center);
\draw[thick, line width=1pt] ([yshift=-.15cm]gb.center) -- ([yshift=.15cm]gb.center);

\draw[decorate,decoration={brace,mirror,amplitude=4pt,raise=.5cm}]
  (ga.south) -- (gb.south)
  node[anchor=north,pos=0.5,yshift=-.55cm,font=\footnotesize]{width $\in [\Delta/4, \Delta]$}
  ;


\node[at={(-0.5,-3)},anchor=east,font=\bfseries]{Step 3};
\draw[-, line width=1pt] (0,-3) -- (10,-3)
  node[name=ha,pos=0.09,anchor=center,label={below:$2^sMF(X-1)$}]{}
  node[name=hb,pos=0.92,anchor=center,label={below:$2^sMF(X)$}]{}
  node[name=hq,pos=0.485,anchor=center,circle,minimum size=5pt,fill=black,inner sep=0pt,label={below:$2^sMQ$}]{}
  ;
\draw[thick, line width=1pt] ([yshift=-.15cm]ha.center) -- ([yshift=.15cm]ha.center);
\draw[thick, line width=1pt] ([yshift=-.15cm]hb.center) -- ([yshift=.15cm]hb.center);

\foreach \x in {0, ..., 20} {
  \node[name=tick-\x, anchor=south,coordinate] at (\x/2,-3){};
  \draw[thick,line width=1pt] (tick-\x.center) -- ([yshift=.25cm]tick-\x.center);
}

\node[name=Um,anchor=south,at=(tick-3.center),label={[label distance=.5mm]above:$\phantom{U}\mathclap{U^-}$}]{};
\node[name=Up,anchor=south,at=(tick-17.center),label={[label distance=.5mm]above:$\phantom{U}\mathclap{U^+}$}]{};
\node[name=V,anchor=south,at=(tick-9.center),label={[label distance=.5mm]above:$V$}]{};

\node[draw=none, fill=cbBlue, opacity=0.4, thick, inner sep=0pt, fit=(Um.south) (Up.north)] {};
\node[ draw=none, fill=cbRed, opacity=0.4, thick, inner sep=0pt, fit=(tick-1) (Um.north)] {};
\node[ draw=none, fill=cbRed, opacity=0.4, thick, inner sep=0pt, fit=(Up.north) (tick-19)] {};

\draw[-latex,black] (hq.north) to[out=90,in=30] (V.center);

\draw[thin,decorate,black,decoration={brace,raise=.275cm}] (ha.center) -- (tick-3.north) node[pos=0,font=\tiny,yshift=0.45cm,anchor=west,inner sep=0pt]{${\le}2$};
\draw[thin,decorate,black,decoration={brace,raise=.275cm}] (tick-17.north) -- (hb.center) node[pos=1,font=\tiny,yshift=0.45cm,anchor=east,inner sep=0pt]{${\le}2$};


\end{tikzpicture}
\caption{Visualization of the randomness recycling mechanism in \cref{alg:online-sampling}.}
\label{fig:demo}
\end{figure}

\noindent
Conditioned on generating $X \sim F$ at a given round,
the conditional distribution of the uniform variate
$Q \sim \Uniform[0,1]$
becomes $Q \mid X \sim \Uniform[F(X-1), F(X)]$.
We will extract a new state $(Z', M')$ that captures unused
randomness from the most recent sampling round.

We first scale the interval $[F(X-1), F(X)]$ by $M 2^s$,
where $s$ is chosen to give an interval width
$2^sM ( F(X) - F(X-1) ) \in [\Delta/4, \Delta]$
(\cref{fig:demo}, Step 2), which implies that
\begin{equation}
(2^sMQ \mid X, M) \sim \Uniform[2^sMF(X-1), 2^sMF(X)].
\label{eq:rhinoscopy}
\end{equation}

We next compute integer points $U^-$ and $U^+$ such that the interval
$\set{ U^-, U^- + 1, \dots, U^+ -1}$ is \emph{strictly} contained in the real interval
$[2^s MF(X-1), 2^sMF(X)]$, with a small gap at the endpoints.
\Cref{eq:rhinoscopy} implies that the floor
$V \defeq \floor{2^sMQ}$ is uniformly distributed over
$\set{U^-, \dots, U^+-1}$ conditioned on falling in this set
(\cref{fig:demo}, Step 3, blue region).
In this case, we extract $Z' \defeq V - U^-$ and $M' \defeq U^+-U^-$ as a
new discrete state $(Z',M')$, with
$Z' \mid M' \sim \Uniform\set{0,\dots,M'-1}$ and $M' \leq \Delta$
as desired.
The factor $2^s M$ is key to efficiently performing these computations given
rational-approximation access to $F$ and bit access to $Q$.

\subsection{Main Results}
Equipped with this background, a high-level summary of our main result is as follows.

\begin{nicebox}[sharp corners]
\begin{theorem}[Informal]
For every online sequence of discrete distributions with computable-real CDF
access, exact online sampling can be performed with $O(\log n)$ total entropy
loss after $n$ samples while storing only $O(\log n)$ bits of persistent state.
Under standard assumptions, the sampler has polynomial expected running time per
sample and $O(\log n)$ expected
overall space over $n$ samples.
The logarithmic overall-space bound is optimal
up to constant factors in the CDF-oracle model.
\end{theorem}
\end{nicebox}

A more formal description of the results is as follows.

\begin{nicebox}[sharp corners,boxrule=-1pt,leftrule=3pt, boxsep=0pt, colframe=black]
\begin{theorem}[Simplified]
\label{thm:rvg-simple}
There is an online random sampling algorithm that,
given CDF oracles for target distributions $F_1, \dots, F_n$,
outputs $X_1 \sim F_1, X_2 \sim F_2, \dots$ with entropy loss
$L_n \leq O(\log n)$,
such that when sampling $X_i$ at iteration $i$,
\begin{itemize}[wide=0pt,leftmargin=*,itemsep=0pt, topsep=1pt]

\item The persistent space is $O(\log i)$ bits;

\item The expected temporary space is $O(\log i + \expect{\log X_i})$ bits;

\item The expected number of calls to the CDF oracle is $O(\expect{\log X_i})$;

\item The expected number of arithmetic operations is $O(\expect{\log X_i})$,
  and the expected bit lengths of the relevant operands are $O(\log i +\expect{\log X_i})$;

\item For efficient CDF oracles (according to the standard notions in~\cref{eq:oracle-accuracy-efficient,eq:oracle-index-efficient}),
the expected running time is within a constant factor of the expected
cost of $\log X_i$ invocations of the oracle at $X_i$ with $\lg X_i$ bits of precision,
plus one invocation of the oracle at $X_i$ with
$\lg i$ bits of precision.
\qedhere
\end{itemize}
\end{theorem}
\end{nicebox}

Our method can also interpolate between entropy loss
$L_n \leq O(\log n)$---with space and runtime per sample scaling as $\log{n}$,
as stated in \cref{thm:rvg-simple}---and entropy loss
$L_n \leq \epsilon n + o_\epsilon(n)$---with space and runtime per sample
scaling as $\log(1/\epsilon)$---yielding a spectrum of trade-offs between
entropy loss and space.
In the latter configuration, the $\log(1/\epsilon)$ persistent space is
asymptotically optimal, matching the lower bound shown in
\citet[Theorem~1]{draper2026lb}.
We further show in \cref{sec:space} that the expected overall space
is essentially optimal, up to constant factors.
The proposed method improves on existing algorithms as follows.

\begin{itemize}
\item The interval algorithm of \citet{han1997} uses exponentially more
persistent and temporary space than our algorithm to achieve a given
entropy loss, as made precise in \cref{table:compare-interval}.
Their algorithm also requires querying the previous CDFs with
increasing precision, while our algorithm does not.

\item The algorithm in \citet{Kozen2022} (for i.i.d.~outputs) generates batched samples,
with entropy loss $L_n \leq \epsilon n + o_\epsilon(n)$ using $O(1/\epsilon)$ persistent space.
Our algorithm uses $O(\log(1/\epsilon))$ space and can support full online
random sampling with non-i.i.d.~outputs.
The algorithms are compared in \cref{table:compare-kozen}.

\item The algorithms in \citet{draper2025rr} assume distributions with
rational probabilities having denominator bounded by $d$, while our
algorithm supports countably many symbols and arbitrary real probabilities.
Further, \citep{draper2025rr} achieve $L_n \leq \epsilon n + o_\epsilon(n)$ using
$O(\log(d/\epsilon))$ space while our algorithm uses $O(\log(1/\epsilon))$,
removing the dependence on the denominator size.
\end{itemize}

\section{Related Work}
\label{sec:related}
As a starting point, \textbf{\citet{knuth1976}} first described
entropy-optimal online sampling, with entropy loss satisfying
$L_n < 2$ for every $n$.
This method presupposes effective access to the bits in the binary expansions of the
output probabilities (and the joint probabilities via cumulative products),
which is possible in the Real RAM model but
incompatible with the standard model of computable reals.
In particular, access to binary expansions is not computably admissible~\citep[Theorem 2.17]{ko1991}
because carries may propagate indefinitely, so basic arithmetic
operations such as addition cannot be implemented bitwise.

\paragraph{\Citet{matias2003}} study online generation of discrete random
variates from a finite set of $N$ dynamically changing rational
weights, with total weight $S$.
The data structures support sampling and
weight updates in a transdichotomous RAM model with
constant-time arithmetic on $O(\log S)$-bit quantities.
In their model, the entropy source emits a uniform integer in $\set{0,\ldots,S-1}$ in
$O(1)$ time; they do not track entropy cost in terms of
individual coin flips.
This setting is similar to ours in its dynamic nature but is otherwise
orthogonal in its model and objectives: \citep{matias2003} target
dynamically changing fixed-precision distributions over finite support without
regard to entropy loss, while we
study entropy-efficient sampling from infinite-support
distributions with CDFs accessed through rational
approximation oracles.

\paragraph{\Citet{han1997}} describe a method of arithmetic coding called
the interval algorithm for online random sampling.
It also handles uniform and nonuniform input and output streams, subject
to an arbitrary dependence structure.
As with our method, it also supports various configurations of the space--entropy
tradeoff.
\Cref{table:compare-interval} compares the two methods under different
choices of parameters.

\begin{table}[H]
\centering
\caption{Complexity and entropy loss trade-offs
for converting i.i.d.~coin flips to random variables
$X_1, \dots, X_n$ subject to any distribution, using the
interval algorithm from \citet{han1997} and \cref{alg:online-sampling}
from this paper (whose properties are summarized in \cref{thm:rvg-simple}).
Here, we define $\mu_i \defeq \expect{X_i}$ (for $i=1,\dots,n$) and $\mu \defeq \max_{1 \le i \le n} \mu_i$.
For simplicity, the table assumes each CDF oracle uses quasilinear time.}
\label{table:compare-interval}
\begin{adjustbox}{max width=\linewidth}
\begin{tabular*}{\linewidth}{|l@{\extracolsep{\fill}}lll|}
\hline
\multirow{2}{*}{\textbf{Method}} & \textbf{Entropy}              & \textbf{Persistent }                           & \textbf{ Expected Time } \\
~                                & \textbf{Loss $L_n$}           & \textbf{Space at Step $i$}                     & \textbf{ Complexity at Step $i$} \\ \hline\hline
\hline
\multicolumn{4}{|@{\;}l|}{\textbf{Bounded Persistent State}} \\
\Citep{han1997}                  &  $\leq \epsilon n + O(1)$      & $O(\textcolor{black}{1/\epsilon} \, \log \mu )$ & $\widetilde{O}\!\left(\textcolor{black}{1/\epsilon} \, \log \mu_i \log \mu \right)$ \\
\Cref{alg:online-sampling}       & $\leq \epsilon n + O(\log n)$ & $O(\textcolor{black}{\log(1/\epsilon)})$        & $\widetilde{O}\!\left( \textcolor{black}{\log(1/\epsilon)}\, \log \mu_i + \log^2 \mu_i\right)$ \\
\hline
\multicolumn{4}{|@{\;}l|}{\textbf{Minimum Entropy Loss}} \\
\Citep{han1997}                  & $O(1)$                        & $O(\textcolor{black}{i} \log \mu)$              & $\widetilde{O}\!\left(\textcolor{black}{i} \log \mu_i \log \mu \right)$ \\
\Cref{alg:online-sampling}       & $O(\log n)$                   & $O(\textcolor{black}{\log i})$                  & $\widetilde{O}\!\left(\textcolor{black}{\log i}\, \log \mu_i + \log^2 \mu_i \right)$ \\
\hline
\end{tabular*}
\end{adjustbox}
\end{table}

Our method requires exponentially less space to
achieve entropy loss $L_n \leq \epsilon n + o_\epsilon(n)$, as $\epsilon$ decreases.
Furthermore, our algorithm can achieve logarithmic space and logarithmic
entropy loss, compared to the linear space and constant entropy loss with
the default iterative approach from \citep[\S{V}]{han1997}.
Note that the persistent space bound of \cref{alg:online-sampling} is
deterministic while the bound \citep{han1997} holds only in expectation.

\paragraph{\Citet{draper2025rr}} describe a general online
random sampling technique
for reducing entropy loss when the
output distributions are over a finite domain
and have rational probabilities with denominator bounded by $d$.
It achieves entropy loss $L_n \leq \epsilon n + O(\log(d/\epsilon))$
for any $\epsilon > 0$ using $O(\log(d/\epsilon))$ persistent
and temporary space.
Randomness recycling is shown to accelerate
\citeauthor{walker1977}'s alias method~\citep{walker1977},
inversion and lookup-table sampling \citep{devroye1986},
DDG tree sampling~\citep{knuth1976,roy2013,saad2020fldr,draper2026a},
and batched uniform sampling~\citep{lemire2019,brackett2025}.
Applying this method to arbitrary discrete distributions
requires each target distribution to be decomposed as a
series of bounded rational distributions, as in \citep[\S4.3]{Kozen2022}.
However, this decomposition is inefficient in general,
requiring a linear scan over the probabilities, with expected time at least linear in the mean of the distribution.
Our algorithm instead manipulates the CDF oracles directly,
allowing logarithmic dependence on the mean via binary search,
while maintaining the same entropy efficiency via a carefully optimized
randomness-recycling procedure.
Thus, we generalize and improve on \citep{draper2025rr}, eliminating the dependence on $d$
and efficiently handling arbitrary distributions.

\paragraph{\Citet{Kozen2022}} discuss an algorithm to produce i.i.d. outputs in batches,
trading off persistent space (i.e., the batch size) with entropy loss.
We improve on \citep{Kozen2022} in two ways.
First, our method handles a dynamic sequence of non-i.i.d.~target
distributions, while the batching method used by \citep{Kozen2022} only
works for an i.i.d.\ output sequence from a fixed distribution.
Second, even for i.i.d.~outputs, our sampling algorithm achieves an
exponential improvement in space.
\Cref{table:compare-kozen} gives the full comparison, showing how the
$1/\epsilon$ overall expected space required by batching in
\citep{Kozen2022} is reduced to $\log(1/\epsilon)$ space
using the randomness-recycling method in this work.

\begin{table}[h]
\centering
\caption{Comparison of overall space complexity
to achieve entropy loss $L_n \leq \epsilon n + o_\epsilon(n)$,
when converting an i.i.d.~input stream of fair bits
to an i.i.d.~output stream with uniform, rational, or arbitrary real probabilities,
using the batched techniques in \citet{Kozen2022} and randomness recycling techniques.}
\label{table:compare-kozen}
\begin{adjustbox}{max width=\linewidth}
\begin{tabular}{|ll||ll|ll|}
\hline
\textbf{Input}     & \textbf{Output} & \bfseries\begin{tabular}[t]{@{}l}Space\\(Batched)\end{tabular} & \textbf{Reference}                        & \bfseries \begin{tabular}[t]{@{}l}Space\\(Recycled)\end{tabular} & \textbf{Reference} \\ \hline\hline
Uniform            & Uniform         & $1/\epsilon$                                                   & \citep[\S4.1]{Kozen2022}                  & $\log(1/\epsilon)$                                               & \citep{jacques2004} \citep[Algorithm 3.1]{draper2025rr} \\
Uniform            & Rational        & $1/\epsilon$                                                   & \citep[\S4.2]{Kozen2022}                  & $\log(1/\epsilon)$                                               & \citep[Algorithms 3.2, 5.1--5.3]{draper2025rr} \\
Uniform            & Real            & $\log n + 1/\epsilon$                                          & \citep[\S4.3]{Kozen2022}                  & $\log(n/\epsilon)$                                               & \cref{alg:online-sampling} \\ \hline
\end{tabular}
\end{adjustbox}
\end{table}

\paragraph{Other Algorithms.}

Several exact sampling algorithms for specific
distributions have been developed, e.g. for the geometric~\citep{Bringmann2013},
the Laplace and Gaussian~\citep{canonne2020,karney2016},
the Bernoulli~\citep{Mendo2025},
and exponential~\citep{thomas2008}.
These methods are tailored to particular distributions
and do not aim to conserve entropy or support online
random sampling from dynamic sequences of distributions.

\Citet{devroye2020} present arbitrary-precision algorithms for
discrete and continuous distributions.
To reduce entropy cost of sampling, they leverage
the interval method~\citep{han1997}
to recycle coin flips, although it is applied in a way that
does not guarantee correct outputs \citep[Footnote 1]{draper2025rr}.

Another direction is exact sampling given a floating-point
approximation to a CDF~\citep{goualard2020,uyematsu2003,saad2025}.
Algorithms to minimize the statistical distance between a finite
discrete distribution with irrational probabilities and a finite rational
approximation are given in \citet{Bocherer2016} using the $\ell_1$
distance and KL divergence, and in \citet{saad2020popl} more generally
using any $f$-divergence.

\paragraph{Complexity Theory.}
\Citet{Feldman1993} develop a probabilistic Turing machine model
for random sampling with bounded (as opposed to expected) polynomial time,
and study the problem of efficiently simulating output dice with
as few different \textit{types} of biased input coins as possible.
\Citet{yamakami1999} studies the relationship between
computable distributions and polynomial-time approximate samplability up to
a small error.
Our oracle model is similar in that the real-valued target probabilities
are accessed via rational approximations.
However, we take computable-real access as the input representation
and ask for exact online sampling with near-optimal entropy
consumption, explicit oracle-cost bounds, and small space.

\Citet{bringmann2017} present various
time and space complexity bounds in the Real RAM for sampling a finite discrete
distribution given unnormalized probabilities
$(p_1, \dots, p_n) \in \mathbb{R}^n_{\ge 0}$.
In this setting, entropy cost is not a traceable aspect of the
complexity analysis, because a uniform real in $[0,1]$---which has
infinitely many random bits---can also be obtained and manipulated in
$O(1)$ time.

\section{Computational Model}
\label{sec:computational-model}

\paragraph{Notation.}
We write $\ln x$ and $\lg x$ for the usual natural logarithm and base-two
logarithm respectively.
Further, we define $\log x \defeq \ln \max\set{ x, e} \geq 1$ for $x \geq 0$.

\paragraph{Word RAM Arithmetic.}

Throughout this paper, we assume a word size of $w \ge 1$ bits and define
$W \defeq 2^w$.
Basic word operations (e.g., addition, subtraction, division,
multiplication) take $O(1)$ time.
A big integer $z$ is stored as a sequence
$(z_{k-1}, \dots, z_0)$ of $w$-bit \emph{limbs}
representing $z = z_{k-1}W^{k-1} + \dots + z_1 W^1 + z_0$.
From $z$ we can form a new integer $z' \defeq z_{k-1}W^{k'-1} + \dots + z_{k-k'}$
by copying the top $k' \le k$ limbs in $O(k')$ time.
A rational number is stored as a pair $(z,y)$ of integers representing $z/y$.
Addition and subtraction on big integers take $k = O(\log_W N)$
word operations.
Other operations are more expensive.
For instance, multiplication and division take
$O(k \log k)$ time using the method of \citet{harvey2021},
or $\Theta(k^2)$ using the schoolbook method
(see \citep[Chapter 15]{gmp2023} for a comprehensive overview).

From a practical standpoint, this model aligns with
existing CPU designs with registers of size $w = 64$.
From a theoretical standpoint, it avoids tracking the bit
complexity of low-level operations.
One important consequence of our representation for big integers
is that it allows \emph{sublinear-time} operations.
The following result will be useful for our algorithm.

\begin{proposition}
\label{prop:approx-div}
There is an algorithm $\Call{ApproxDivision}{q,k}$ that takes as
input a rational number $q = a/b \in [0,1]$ and integer $k \ge 1$,
and returns in $O(k)$ time a rational number $q' = a'/b'$
such that $\abs{q - q'} < 2^{-k}$,
where $0 \le a' \le b' < 2^{k+1}$
(i.e., $a'$ and $b'$ each have bit length at most $k+1$).
\end{proposition}

\begin{proof}
See \cref{sec:approx-div}.
\end{proof}

\paragraph{Representing Probability Distributions.}
We consider probability
distributions with support $\Nat$.
A CDF $F: \Nat \to \Real$
is a nondecreasing
function with $F(0) \geq 0$ and $\lim_{i \to \infty} F(i)=1$.
The corresponding PMF is given by
$f(i) \defeq F(i) - F(i-1)$ for $i \geq 0$,
where by convention we write $F(-1) \defeq 0$.

An \textit{oracle} for $F$ is a
function $\phi: \Nat \times \Nat \to \Rat$ such that
$\abs{F(i) - \phi(i,k)} < 2^{-k}$, and
$\phi(i,k)$ has denominator $2^{k}$ for every $i,k$.
We also define $\phi(-1,k) \defeq 0$ for all $k$.

We define the \textit{bound oracles}
\begin{align}
\Fl{i}{k} \defeq \phi(i,k+1) - 2^{-(k+1)}
\qquad \text{and} \qquad
\Fu{i}{k} \defeq \phi(i,k+1) + 2^{-(k+1)},
\end{align}
which satisfy
\begin{align}
\Fl{i}{k} < F(i) < \Fu{i}{k}
\qquad \text{and} \qquad
\Fu{i}{k} - \Fl{i}{k} \leq 2^{-k}.
\label{hrup5}
\end{align}

To account for the cost of evaluating the oracle,
we define $\CF{x}{\nu}$:  $x \geq 0, \nu \geq 1$  to be an upper bound on the number
of word-level operations needed to evaluate the oracle
$\phi(i',k')$ over all $i' \leq \ceil{x}$, $k' \leq \ceil{\lg \nu} + 1$, and we suppose that
$\mathcal{C}_F$ is nondecreasing in $x$ and $\nu$.
We say the oracle is \textit{accuracy efficient} /\textit{index efficient}
if there is a polynomial $g$ with
\begin{align}
&\mbox{\bfseries \underline{accuracy efficient}} & \frac{\CF{x}{\nu'}}{\CF{x}{\nu}} &\le g\!\left( \frac{\log \nu'}{\log \nu} \right) & (x \geq 0; 1 \leq \nu < \nu'), \label{eq:oracle-accuracy-efficient}  \\
&\mbox{\bfseries \underline{index efficient}}    & \frac{\CF{x'}{\nu}}{\CF{x}{\nu}} &\le g\!\left( \frac{\log x' }{\log x} \right) & (\nu \geq 1; 0 \leq x < x'). \label{eq:oracle-index-efficient}
\end{align}
An oracle for $F$ is \emph{efficient} if it is both
accuracy-efficient and index-efficient.
Many distributions have efficient oracles; for instance, the Poisson
distribution (via algorithms to estimate the incomplete Gamma
functions~\citep{winitzki2003computing}) or the discretized Gaussian
distribution (via algorithms to estimate the error
function~\citep{chevillard2012functions}).

\begin{remark}[name={Note on Complexity Bounds}]
When dealing with efficient oracles, all the asymptotic notations may have
a hidden dependence on the polynomial $g$ in
\cref{eq:oracle-accuracy-efficient,eq:oracle-index-efficient}.
The intent is that, for a certain genre of oracle (e.g. oracles for
Gaussians), the terms in $g$ would be small fixed constants.
To avoid edge cases in the analysis, we set
$\CF{x}{\nu} \defeq \CF{x}{1}$ for $\nu < 1$, and
we assume that $\CF{x}{\nu}$ is at least as large as the cost of basic arithmetic
operations on quantities of magnitude $\ceil{x}$ and $\nu$.
This convention ensures that the cost of computing $F^{\pm}(i,k)$
is $O\!\left(\CF{i}{2^k}\right)$ for $k \geq 0$.
\end{remark}

The next result establishes the cost of estimating the PMF given a CDF oracle.

\begin{proposition}
\label{prop:estimate-pdf}
There is an algorithm $\Call{EstimatePMF}{F, x, \nu}$
that takes a CDF $F$, an integer $x \geq 0$, and an integer $\nu \geq 1$,
and returns an estimate $\hat p$ of the probability
mass $p \defeq f(x)$
such that $\hat p \le p \le (1+1/\nu) \hat p$,
with runtime $O\!\left( \log(\nu/p)\thinspace \CF{x}{8 \nu/p} \right).$
(If $p = 0$, the algorithm runs forever without returning any output).
\end{proposition}

\begin{proof}
See \cref{sec:pmf}.
\end{proof}

\section{Online Random Sampling: Overview}
\label{sec:basic-online-sampling}

As a starting point, \cref{alg:basic-sampling} shows a simple, stateless
procedure for random sampling.
It is based on the \citet{han1997} interval method adapted to the case that
the target CDF is given via an oracle.

\begin{algorithm}[H]
\caption{Online random sampling using the (stateless) \citet{han1997} interval method.}
\label{alg:basic-sampling}
\begin{algorithmic}[1]
\Require{Oracle for CDF $F$}
\Ensure{Random sample $X \sim F$}
\Procedure{IntervalSampleStateless}{$F$}
\State Initialize $r \gets 0$, $\varlo \gets 0$, $\varhi \gets \infty$, $\hat Q \gets 0$.
\While{$\varlo < \varhi$}
  \State $\varmid \gets \begin{cases}
    2 \varlo  & \text{if $\varhi = \infty$} \\
    \floor{(\varhi+\varlo)/2} & \text{otherwise}
  \end{cases}$

  \If{$\Fu{\varmid}{r} \leq \hat Q$}
    \State $\varlo \gets \varmid+1$
  \ElsIf{$\hat Q + 2^{-r} \leq \Fl{\varmid}{r}$}
    \State $\varhi \gets \varmid$
  \Else
    \State $r \gets r+1$
    \State $B_r \gets \mbox{fresh random bit}$ \State $\hat Q\gets \sum_{i=1}^r B_i 2^{-i}$
  \EndIf
\EndWhile
\State \Return $\varlo$
\EndProcedure
\end{algorithmic}
\end{algorithm}

This sampling algorithm performs a dichotomous search over
$\Nat$ to find the value $X$ with $Q \in [F(X-1), F(X)]$
where $Q \sim \Uniform[0,1]$ is an infinitely precise real
that is known to lie in the  interval $[\hat Q, \hat Q + 2^{-r}]$.
The oracle for $F$ is queried at the current midpoint $\varmid$ with $r$ bits of precision.
If $\hat Q$ is to the right of the upper bound, we recurse right.
If $\hat Q + 2^{-r}$ is to the left of the lower bound, we recurse left.
If neither condition holds, then a fresh bit is sampled to narrow the
interval containing $Q$.

\begin{proposition}
\label{prop:basic-sampling-complexity}
\Cref{alg:basic-sampling} returns a sample from distribution $F$, and it
uses $H(F) + O(1)$ random bits in expectation. Moreover, if $F$ is
efficient, then the expected runtime is
\begin{equation}
O\!\left(\expect{ \log(X) \thinspace \CF{X}{X} } \right).
\end{equation}
\end{proposition}

\begin{proof}
This is a direct corollary of the analysis in
\cref{sec:online-sampling-single}
for the special case $\Delta=M=1$, where \cref{alg:basic-sampling} and
\cref{alg:online-sampling} are essentially equivalent.
\end{proof}

While \cref{alg:basic-sampling} is efficient in runtime, it
wastes significant entropy. In the remainder of this paper, we present our novel online random
sampling method. The basic plan, given in \cref{alg:loop}, is to repeatedly call a method \Call{Sample}{}
(\cref{alg:online-sampling}).
The \Call{Sample}{} method generates a single
variable $X$ given a CDF $F$ using i.i.d.~coin flips; it additionally consumes
and produces a discrete uniform random state
$(Z,M)$ of recycled randomness,
where $Z \mid M \sim \Uniform\set{0,\dots,M-1}$ and $M \leq \Delta$.

\begin{algorithm}[H]
\caption{Online random sampling with randomness recycling.}
\label{alg:loop}
\begin{algorithmic}[1]
\Require Entropy loss parameter $\epsilon \in [0,1]$
\Ensure Exact online random sampling with entropy loss $L_n \leq \epsilon n + 11\lg(n+2)$
  and $O(\log\min\set{n,1/\epsilon})$ bits persistent space, after any $n \ge 1$ iterations
\State $t \gets 0$, $\Delta \gets 0$, $(Z, M) \gets (0, 1)$
\While{true} \label{algline:online-sampling-loop}
\If{$\epsilon\, t < 0.1 $}
  \State $t \gets t+1$
  \State $\Delta \gets 5000 t \lceil \lg(t+2) \rceil$
  \EndIf
  \State \textbf{Receive} CDF $F$ of next target distribution $F$
  \State $(X, (Z,M)) \gets \Call{Sample}{F, (Z, M), \Delta}$
    \Comment{call \cref{alg:online-sampling}}
  \State \textbf{Emit} $X$
\EndWhile
\end{algorithmic}
\end{algorithm}

With this parametrization, the values after each iteration $k$ satisfy
\begin{equation}
0 \le Z < M \leq \Delta \leq \poly(t) \leq \poly( \min\set{1/\epsilon, k} ).
\end{equation}

The total persistent space to store the global state of \cref{alg:loop}
is therefore $O\!\left(\log \min\set{1/\epsilon, k} \right)$.
We will establish the following main result.
\begin{theorem}[store=entropy-bound-overall,restate-keys={note={Restated}}]
\label{thm:entropy-bound-overall}
The entropy loss $L_n$ after $n$ iterations
of the sampling loop in \cref{alg:loop} satisfies
\begin{equation}
L_n \leq \epsilon n + 11 \lg(n+2).
\end{equation}
\end{theorem}

\Cref{alg:online-sampling} provides the main individual sampling procedure \Call{Sample}{}. This procedure -- in particular, the subroutine \Call{Recycle}{} -- implements the rule described
previously in \cref{fig:demo}.
We describe \Call{Sample}{} in \cref{sec:online-sampling-single}, and
analyze the sampling loop from \cref{alg:loop} in \cref{sec:online-sampling-overall}.

\begin{algorithm}[!t]
\caption{Single iteration of online random sampling with randomness recycling.}
\label{alg:online-sampling}

\begin{algorithmic}[1]
\Require{%
  \begin{tabular}[t]{@{}l}
    Oracle for CDF $F$  \\
    Random state $(Z,M)$ such that $Z \mid M \sim \mathrm{Unif}\set{0,\dots,M-1}$ \\
    Parameter $\Delta \geq M$
  \end{tabular}}
\Ensure{
  \begin{tabular}[t]{@{}l}
    Random sample $X \sim F$ \\
    Random state $(Z',M')$ such that
    $Z' \mid M' \sim \mathrm{Unif}\set{0,\dots,M'-1}$ and $M' \leq \Delta$
  \end{tabular}}
\Procedure{Sample}{$F, (Z, M), \Delta$}
\State $(X, (B_1, \dots, B_r)) \leftarrow \Call{IntervalSample}{F,(Z,M)}$
\State $(Z', M') \leftarrow \Call{Recycle}{F, (Z,M), \Delta, X, (B_1, \dots, B_r)}$
\State  \Return $(X, (Z', M'))$
\EndProcedure

\bigskip

\Procedure{IntervalSample}{$F, (Z, M)$}
  \State $r \gets 0, t \gets 0, \varlo \gets 0, \varhi \gets \infty, \hat Q \gets 0$
  \While{$\varlo < \varhi$}
  \State $\varmid \gets \begin{cases}
  2 \varlo & \text{if $\varhi = \infty$} \\
  \floor{(\varhi+\varlo)/2} & \text{otherwise}
  \end{cases}
  $
  \If{$\Fu{\varmid}{t} \leq \hat Q$}
    \State $\varlo \leftarrow \varmid + 1$
  \ElsIf{$\hat Q + 2^{-t+1} \leq \Fl{\varmid}{t}$}
    \State $\varhi \leftarrow \varmid$
  \Else
    \State $t \leftarrow t+1$
    \If{$t \geq 1 + \lfloor \lg M \rfloor$}
    \State $r \leftarrow r+1$
      \State $B_{r} \gets \mbox{fresh random bit}$ \label{algline:online-sampling-bit1}
      \State $\hat Q \gets \left(Z + \sum_{i=1}^{r} B_i 2^{-i}\right)/{M}$
      \label{qform1}
      \Else
  \State $\hat Q \leftarrow \Call{ApproxDivision}{Z/M,t+1} - 2^{-t-1}$ \label{qform2}
\EndIf
    \EndIf
  \EndWhile
  \State \Return $(\varlo, (B_1, \dots, B_r))$
\EndProcedure

\bigskip

\Procedure{Recycle}{$F, (Z, M), \Delta, X, (B_1, \dots, B_r)$}
  \State $\hat p \leftarrow \Call{EstimatePMF}{F, X, 1}$
  \State $s \gets \max \set*{ 0, \ceil*{\lg  \frac{\Delta}{M\hat p}} - 2}$
  \If{$r < s$}
    \State $(B_{r+1}, \dots, B_s) \gets \mbox{fresh random bits}$ \label{algline:online-sampling-bit2}
  \EndIf
  \State $U^- \gets \left\lceil 2^s M \Fu{X-1}{s + \lceil \lg M \rceil} \right\rceil$
  \State $U^+ \gets \left\lfloor 2^s M \Fl{X}{s + \lceil \lg M \rceil} \right\rfloor$
  \State $V \gets 2^s\!\left(Z + \sum_{i=1}^s B_i 2^{-i} \right)$
  \If{$U^- \leq V < U^+$}
    \State \Return $(V - U^-, U^+ - U^-)$
  \Else
    \State \Return $(0, 1)$
  \EndIf
\EndProcedure
\end{algorithmic}
\end{algorithm}

\section{Single Variable Sampling with Uniform Random States}
\label{sec:online-sampling-single}
Here $\Delta \ge M$ is an integer parameter dictating
the size of the output state. We will analyze the call
\begin{equation*}
(X, (Z', M')) \leftarrow \Call{Sample}{F,(Z,M),\Delta}.
\end{equation*}

Fix $M$, and suppose that $Z \sim\Uniform\set{0, \dots, M-1}$
and $B \defeq (B_1, B_2, \dots)$ is
an infinite sequence of random bits from the entropy source.
Then the following quantity is uniformly distributed over $[0,1]$:
\begin{equation}
Q = \frac{Z + \sum_{i=1}^{\infty} B_i 2^{-i}}{M}.
\label{eq:q-uniform}
\end{equation}

\Call{Sample}{} begins by calling
$\Call{IntervalSample}{F,(Z,M)}$,
which uses the randomness in $Z$ to extend the baseline method $\Call{IntervalSampleStateless}{F}$.
The call to \Call{IntervalSample}{} returns a sample
$X \sim F$ as well as the random bits $(B_1, \dots, B_r)$
used to generate the outcome.
The interval $[\hat Q, \hat Q + 2^{-t+1}]$ approximates the real quantity $Q$.
Next, the call to \Call{Recycle}{} attempts to recycle the used random
bits into a new random state $(Z', M')$, which may require fresh
random bits to obtain a wide enough discrete interval around $Q$.

\begin{proposition}
\label{obs:interval-sample-invariants}
The following invariants hold in each iteration of \Call{IntervalSample}{}:
\begin{align}
r=\max &\set{0,t-\floor{\lg M}}, \label{hbg0} \\
\hat Q      \leq{} & Q \leq \hat Q+2^{-t+1}, \label{hbg1}
  \\
F(\varlo-1) \leq{} & Q \leq F(\varhi). \label{hbg2}
\end{align}
\end{proposition}

\begin{proof}
We show \cref{hbg0} by induction on $t$.
The base case holds since $t = r = 0$ and $M \geq 1$.
For the induction step, if the new value has
$t < 1 + \lfloor \lg M \rfloor$, then
$t - \lfloor \lg M \rfloor \leq 0$ and $r$ remains zero.
If the new value has $t \geq 1 + \lfloor \lg M \rfloor$, then both $t$ and
$r$ increment, while $t - \lfloor \lg M \rfloor$ also increments and is
strictly positive.

For \cref{hbg1}, there are two cases.
Suppose $r > 0$ and $\hat Q$ is derived at \cref{qform1}.
So $\hat Q \leq Q \leq \hat Q + 2^{-r}/M$, as the missing bits
$B_{r+1}, \dots$ can contribute at most $2^{-r}$ in total.
By \cref{hbg0}, we have
$2^{-r}/M = 2^{-t + \lfloor \lg M \rfloor} / M \leq 2^{-t} \leq 2^{-t+1}$.
Otherwise, suppose $r = 0$ and $\hat Q$ is derived at \cref{qform2}.
Then we have $Z/M \leq Q \leq (Z+1)/M$.
By the specification of \Call{ApproxDivision}{}, we have
$Z/M - 2^{-t} \leq \hat Q \leq Z/M$
and
$Q - \hat Q \leq (Z+1)/M - (Z/M - 2^{-t}) \leq 1/M + 2^{-t}$.
Since $r = 0$, we have $t \leq \lfloor \lg M \rfloor$ and so $1/M \leq 2^{-t}$.

For \cref{hbg2}, first note that $F(-1) = 0 \leq Q \leq 1 = F(\infty)$ initially.
Next, observe that if $\varlo$ updates to $\varmid+1$,
then by \cref{hbg1} we have $F(\varmid) \leq \Fu{\varmid}{t} \leq \hat Q \leq Q$.
So $F(\varlo - 1) \leq Q$.
Similarly, if $\varhi$ updates to $\varmid$, then by \cref{hbg1} we have
$Q \leq \hat Q + 2^{-t+1} \leq \Fl{\varmid}{t} \leq F(\varmid)$.
\end{proof}

\begin{proposition}
\label{prop:recycling-integer-bounds}
The following invariants hold for $\Call{Recycle}{F, (Z,M), \Delta, X, (B_1, \dots, B_r)}$:
\begin{eqnarray}
\Delta/4       \leq & 2^s M f(X) & \leq \Delta \label{hag0} \\
2^s M F(X) - 2 \leq & U^+        & \leq 2^s M F(X) \label{hag1} \\
2^s M F(X-1)   \leq & U^-        & \leq 2^s M F(X-1) + 2 \label{hag2} \\
V              \leq & 2^s M Q    & \leq V + 1. \label{hag3}
\end{eqnarray}
\end{proposition}
\begin{proof}
By specification of \Call{EstimatePMF}{},
we have $\hat p \leq f(X) \leq 2 \hat p$.

The bounds \cref{hag0} hold from the specification of $\hat p$, the
definition of $s$, and our hypothesis that $M \leq \Delta$.

The bounds \cref{hag1} hold because
\begin{equation}
\frac{U^+}{2^s M} \leq \Fl{X}{s + \ceil{\lg M}} \leq F(X) \leq \Fl{X}{s + \ceil{\lg M}} + \frac{1}{2^s M} \leq \frac{U^+ + 2}{2^s M}.
\end{equation}
The bounds \cref{hag2} are completely analogous.

The bounds \cref{hag3} hold since
$2^s M Q - V = 2^s \sum_{i=s+1}^{\infty}B_i 2^{-i}$ and
$0 \leq 2^s \sum_{i=s+1}^{\infty}B_i 2^{-i} \leq 1$.
\end{proof}

\begin{proposition}
\label{prop:interval-search}
The sampled value $X$ is the unique index with $Q \in [F(X-1),F(X)]$
(aside from the probability-zero event that $Q$ equals a CDF boundary),
and $X$ is $F$-distributed.
Moreover, if we condition on the event $\set{X = x}$ for
any $x$ with $f(x) > 0$,
then $Q$ is uniformly drawn from the range $[F(x-1), F(x)]$.
\end{proposition}

\begin{proof}
Because $\varlo = \varhi = X$ at termination, we have $Q \in [F(X-1), F(X)]$
by \cref{obs:interval-sample-invariants}.
Because $Q$ is uniformly distributed in $[0,1]$,
$\Pr(X = x) = \Pr(Q \in [F(x-1), F(x)]) = f(x)$,
and $Q$ is uniform in $[F(x-1), F(x)]$ conditioned on $X = x$.
\end{proof}

\begin{proposition}
\label{prop:recycling-correctness}
The recycled values $(Z', M')$ satisfy $0 \leq Z' < M' \leq \Delta$ always.
Furthermore, conditioned on the values of $X$, $M$, and $M'$, the output
$Z'$ is uniformly distributed over $\set{0, \dots, M' - 1}$.
\end{proposition}

\begin{proof}

When we condition on $M$ and $X = x$, the values $s$, $U^-$, and $U^+$ are completely
determined, because $F$ is a deterministic oracle.

The results are trivial if $M' = 1$, so consider the case with
$U^- \leq V < U^+$ and $M'=U^+-U^-$.
Here, from \cref{prop:recycling-integer-bounds}, we have $U^+ - U^- \leq
2^s M F(X) - 2^s M F(X-1) = 2^s M f(X) \leq \Delta$.
Furthermore, for any choice of variables $(Z, B_1, \dots, B_s)$ which satisfy
$U^- \leq V < U^+$, we have $Q \in [F(x-1), F(x)]$
and thus $X = x$.
Therefore, all such values of $(Z, B_1, \dots, B_s)$ will still have a uniform
distribution conditioned on $X = x$; in particular, $V - U^-$
is uniform in the range $\set{0, \dots, U^+ - U^- - 1}$.
\end{proof}

\begin{proposition}
\label{prop:precision-geometric-tail}
Let $k \geq 0$.
The final value of the variable $t$ in
$\Call{IntervalSample}{F, (Z,M)}$ satisfies
\begin{equation}
\Pr\left( t > k \mid X, M\right) \leq \frac{6}{2^k f(X)}.
\end{equation}
\end{proposition}

\begin{proof}
Consider the iteration $t=k$ just before $t$ increments to $k+1$.
As neither of the conditions modifying $\varlo$ or $\varhi$ held then, we  have
\begin{equation}
\Fl{\varmid}{t}  - 2^{-t+1} \leq \hat Q \leq \Fu{\varmid}{t}
\end{equation}

As $\Fu{\varmid}{t}  \leq F(\varmid) + 2^{-t}$ and
$\Fl{\varmid}{t} \geq F(\varmid) - 2^{-t}$, we have
\begin{equation}
\label{eq:fa1}
\hat Q - 2^{-t} \leq F(\varmid) \leq \hat Q + 3 \cdot 2^{-t}.
\end{equation}

Putting together \cref{eq:fa1,hbg1} gives
\begin{equation}
\label{eq:q1r}
Q - a \leq F(\varmid) \leq Q + a \qquad \text{where $a = 3 \cdot 2^{-k}$}
\end{equation}

Now suppose we condition on $X = x$. If $\varmid \geq x$, then \cref{eq:q1r} implies that
$Q \geq F(\varmid) - a \geq F(x) - a$; likewise if $\varmid < x$ then \cref{eq:q1r}
implies that $Q \leq F(\varmid) + a \leq F(x-1) + a$.
Putting the two cases together, we conclude that
\begin{equation}
\label{eq:q2r}
Q \in \left[ F(x-1), F(x-1) + a \right] \cup \left[F(x) - a, F(x) \right].
\end{equation}

By \zcref{prop:interval-search}, the probability that \cref{eq:q2r}
holds conditional on $X = x$ is at most
$\frac{2 a}{f(x)}$.
\end{proof}

\subsection{Analysis of Entropy Loss}
\label{sec:online-sampling-entropy}

\begin{theorem}
\label{thm:entropy-bound-conditional-surprisal}
If $\Delta \geq 1000$, then the total number of sampled bits $T \defeq \max\set{r, s}$ in
\cref{algline:online-sampling-bit1,algline:online-sampling-bit2} of \cref{alg:online-sampling} satisfies
\begin{equation}
\expect{T - \lg M' \mid M, X} \leq \lg(1/f(X)) - \lg M + \frac{48 \lg \Delta}{\Delta}.
\end{equation}
\end{theorem}

\begin{proof}
Conditioned on $X = x$, the value $s$ is a constant while $r$
remains a random variable. We have
\begin{equation}
\expect{T- \lg M' \mid M, X }
  = \left( s + \sum_{k=s}^{\infty} \Pr( r > k) \right)
    - \left( \Pr\left( U^- \leq V < U^+\right) \lg\left(U^+ - U^-\right) \right).
\end{equation}

Let $\theta \defeq 2^s M f(x)$.
By \cref{hbg0} and \cref{prop:precision-geometric-tail}, we have
\begin{align}
\sum_{k=s}^{\infty} \Pr( r > k )
= \sum_{k=s}^{\infty} \Pr( t > k + \lfloor \lg M \rfloor)
\leq \sum_{k = s}^{\infty} \frac{6}{2^{k + \lfloor \lg M \rfloor} f(x)}
= \frac{12}{2^{s + \floor{\lg M}} f(x)}
\leq  \frac{24}{\theta}.
\end{align}

Next, the bound \cref{hag3} implies that
\begin{align}
\Pr ( U^- \leq V < U^+ \mid X )
= \frac{U^+ - U^-}{2^s M f(x)}
= \frac{U^+ - U^-}{\theta}.
\end{align}

By \cref{hag1} and \cref{hag2} we have $U^+ - U^- \geq \theta - 4$, so overall
\begin{align}
\expect{T - \lg M' \mid M, X}
&\leq s + \frac{24}{\theta} - \left( 1 - \frac{4}{\theta} \right) \lg( \theta - 4) = s - \lg  \theta + \frac{24}{\theta} - \lg (1 - 4/\theta) + \frac{4}{\theta} \lg (\theta-4)
\end{align}

Here, $s - \lg \theta = \lg(1/f(x)) - \lg M$. Note that $\theta\ge \Delta/4\ge 250$. Using this fact and straightforward calculations, we get
\begin{equation}
\label{eq:xccs1}
\frac{24}{\theta} - \lg \left(1 - \frac{4}{\theta}\right) + \frac{4}{\theta} \lg (\theta - 4)
\leq \frac{12 \lg \theta}{\theta}
\le \frac{48\lg\Delta}{\Delta}. \qedhere
\end{equation}
\end{proof}

\subsection{Analysis of Complexity}
\label{sec:online-sampling-runtime}

\begin{proposition}
\label{prop:runtime-conditional-general}
Conditioned on $X = x$, the expected runtime of \Call{Sample}{} is
\begin{equation}
O\!\left(\CF{x}{\Delta/p} + \log( x+1/p) \sum_{i=0}^{\infty} 2^{-i} \CF{2x}{2^i/p} \right)
\qquad \text{for $p = f(x)$}.
\end{equation}
\end{proposition}

\begin{proof}
In the phases of \Call{IntervalSample}{} where $\varhi = \infty$, we always have $\varlo \leq x$;
in each such iteration, we set $\varmid = 2 \varlo \leq 2 x$.
We can only modify $\varlo$ this way $O(\log x)$ times.
Then, after $\varhi < \infty$, each iteration that modified $\varhi$ or
$\varlo$ will reduce $\varhi-\varlo$ by a constant factor.
Again, such iterations can occur only $O(\log x)$ times.
Overall, $\varmid$ takes $O(\log x)$ different values,
all of which are at most $2 x$.

We claim that each iteration of \Call{IntervalSample}{} has cost
$O(\CF{2x}{2^{t}})$. Indeed, if $r = 0$, then the numerator and denominator of $\hat{Q}$ are
$\poly(2^t)$, and we query $F^{\pm}(\varmid, t)$, so the arithmetic is
performed on values of magnitude polynomial in $2^t$.
Even though $(Z, M)$ may have bit length $O(\log \Delta)$, we avoid any
dependence on $\Delta$ via \cref{prop:approx-div}.
If $r > 0$, then
$M < 2^t$ and again the arithmetic is performed on values of size
$2^t M \leq \poly(2^t)$.

Let $t_0 = \lfloor \lg 1/p \rfloor$. There are at most $O( t_0 + \log x) =
O(\log(x + 1/p))$ iterations with $t \leq t_0$, since in each iteration
either $\varmid$ or $t$ is updated;
each such iteration has cost $O(\CF{2x}{2^t}) \leq O(\CF{2x}{1/p})$.
So the total cost of all such iterations is at most
\begin{equation}
\label{bar1}
O\!\left( \log(x+1/p) \thinspace \CF{2x}{1/p}\right).
\end{equation}

Likewise, for $i > 0$,  there are at most $O(\log x)$ iterations with
$t = t_0 + i$, and the cost of each such iteration is
$O(\CF{2x}{2^i/p})$.
By \cref{prop:precision-geometric-tail}, the probability of reaching such
value $t$ is at most $1/(2^t p) \leq O( 2^{-i} )$.
So the total contribution from these iterations is at most
\begin{equation}
\label{bar2}
O\left( \log(x) \thinspace 2^{-i} \CF{2x}{2^i/p} \right)
\end{equation}

The contributions from \cref{bar1} and \cref{bar2}, summed over all $i > 0$,
together can be bounded as
\begin{equation}
\label{bar3}
O \Bigl( \log(x+1/p) \thinspace \sum_{i=0}^{\infty} 2^{-i} \CF{2 x}{2^i/p} \Bigr).
\end{equation}

In \Call{Recycle}{}, the call \Call{EstimatePMF}{} has cost
$O(\log (1/p) \thinspace \CF{x}{8/p})$, which
is within a constant factor of \cref{bar3}.
The arithmetic for $U^-$ and $U^+$ and oracle queries for $F(x-1)$
and $F(x)$ have cost $O(\CF{x}{\Delta/p})$.
\end{proof}

\begin{theorem}
\label{thm:accuracy-efficient}
Suppose the oracle is accuracy efficient.
Conditioned on $X = x$, the expected runtime of
$\Call{Sample}{F, (Z,M), \Delta}$ is
\begin{equation}
O \left(\CF{x}{\Delta} + \log( x+1/p) \thinspace \CF{2x}{1/p} \right)
\qquad \text{for $p = f(x)$}.
\end{equation}
\end{theorem}

\begin{proof}
By \cref{prop:runtime-conditional-general}, the expected runtime, ignoring
constant factors, is given by
\begin{equation}
\CF{x}{\Delta/p }
  + \log( x+1/p) \sum_{i=0}^{\infty} 2^{-i} \thinspace \CF{2x}{2^i/ p}.
\end{equation}

Using our assumption on $\cal C_F$, we can bound the sum $\sum_{i=0}^{\infty} 2^{-i} \thinspace \CF{2x}{2^i/ p}$ by
\begin{align}
O(\CF{2x}{1/p}) \sum_{i=0}^{\infty} 2^{-i} g\!\left(\frac{\log(2^i/p)}{\log(1/p)}\right)
\leq O(\CF{2x}{1/p}) \sum_{i=0}^{\infty} 2^{-i} g(i+1) \leq O(\CF{2x}{1/p})
\end{align}
using the polynomial $g$ from \cref{eq:oracle-accuracy-efficient}. We can also bound
\begin{align}
\CF{x}{\Delta/p}
&\leq \CF{x}{1/p^2} + \CF{x}{\Delta^2} \leq \CF{x}{1/p} g\!\left(\frac{\log(1/p^2)}{\log(1/p)}\right)
    + \CF{x}{\Delta} g\!\left(\frac{\log(\Delta^2)}{\log(\Delta)}\right)
 \\
&\leq O(\CF{x}{1/p} + \CF{x}{\Delta}).
\end{align}

The term $ \CF{x}{1/p}$ is dominated by the other term $\CF{2x}{1/p}$
appearing in the stated bound.
\end{proof}

\begin{theorem}
\label{thm:efficient-runtime}
If the oracle is efficient, then the expected runtime of
$\Call{Sample}{F, (Z,M), \Delta}$ is
\begin{equation}
O \Bigl( \expect[X \sim F]{ \CF{X}{\Delta} + \log (X) \thinspace \CF{X}{X} } \Bigr).
\end{equation}
\end{theorem}

\begin{proof}
By \cref{thm:accuracy-efficient}, the expected runtime is, up to constant
factors, given by
\begin{equation}
\begin{aligned}
\expect{\CF{X}{\Delta}} + \expect{\log(X+1/f(X)) \thinspace \CF{2 X}{1/f(X)}},
\end{aligned}
\end{equation}
where $X$ follows the distribution $F$.
Consider the function
\begin{equation}
h(x,y) = \CF{x}{\Delta} + \log(x+y)\thinspace \CF{2 x}{y}.
\end{equation}
It can be seen that this function satisfies the condition
\begin{align}
\frac{h(x,y)}{h(x',y')} \leq C\!\left( \frac{\log x \log y}{\log x' \log y'} \right)^k
&& (0 \leq x' \leq x, 0 \le y' \leq y)
\end{align}
for some constants $C, k$.
As we show in \cref{lemma:ff} from \cref{sec:lemmas}, such a function satisfies
\begin{equation}
\expect{h(X, 1/f(X))} \leq O(\expect{h(X, X)})
\end{equation}
Here $h(X, X) \leq O( \CF{X}{\Delta} + \log(X) \thinspace \CF{X}{X})$.
\end{proof}

\begin{corollary}
\label{cor:efficient-runtime-single-sample}
Let $\mu = \expect{X}$.
If the oracle is efficient, then $\Call{Sample}{F, (Z,M), \Delta}$
has expected runtime
\begin{equation}
O \Bigl( \CF{\mu}{\Delta} +  \log(\mu) \thinspace \CF{\mu}{\mu} \Bigr).
\end{equation}
\end{corollary}
\begin{proof}
By monotonicity and efficiency,
for every $x \in \Nat$ we have
\begin{equation}
\log(x) \thinspace \CF{x}{x}
\le
\log(\mu) \thinspace \CF{\mu}{\mu} \thinspace g(x),
\qquad
g(x) \defeq C\!\left(1+\left(\frac{\log x}{\log\mu}\right)^k \right).
\end{equation}

It remains to show that $\expect{g(X)}=O(1)$.
By our definition of the $\log$ function and Jensen's inequality,
\begin{equation}
\log^k x \leq \log^k\left(x + e^k\right)
\implies
\expect{\log^k X} \le \expect{\log^k\left(X + e^k\right)} \le \log^k\left(\mu + e^k\right) = O\!\left(\log^k \mu\right).
\end{equation}
Therefore $\expect{g(X)}=O(1)$.
The same argument, without the
leading factor $\log x$, applies to $\CF{x}{\Delta}$.
The claim follows from \cref{thm:efficient-runtime}.
\end{proof}

\begin{proposition}
\label{prop:sgeomtail}
Let $S$ denote the temporary space of $\Call{Sample}{F, (Z,M), \Delta}$.
Conditioned on $X$, there is some constant $c$ with
\begin{equation}
\Pr( S \geq c \log (X + \Delta/f(X)) + \lambda \mid X) \leq O( 2^{-\lambda})
\end{equation}
for all $\lambda \geq 1$.  In particular,
\begin{equation}
\expect{S \mid X} \leq O( \log(X + \Delta/f(X)) ).
\end{equation}
\end{proposition}
\begin{proof}
The space of \Call{IntervalSample}{} is $O(t + \log (X + \Delta))$,
and the space of \Call{Recycle}{} is $O(\log(\Delta / f(X)))$.
By \cref{prop:precision-geometric-tail}, conditioned on $X$, the final
value of $t$ has mean at most $O(\log(1/f(X)))$, with the tail bounded by a
geometric distribution.
\end{proof}

\section{Online Random Sampling Analysis}
\label{sec:online-sampling-overall}

We are now prepared
to analyze the main sampling loop of
\cref{alg:loop}.
We index iterations of the loop in \cref{alg:loop} starting from $i=1$;
each iteration $i$ receives $F_i$, starts from state $(Z_{i-1},M_{i-1})$,
calls \Call{Sample}{} with parameter $\Delta_i$, and returns $X_i$ together
with the new state $(Z_i,M_i)$.
We set $(Z_0,M_0) \defeq (0,1)$.

\begin{proposition}
\label{prop:entropy-loss}
For every $n$,  the entropy loss after the first $n$ iterations
of \cref{alg:loop} satisfies
\begin{equation}
    L_n \le \lg\Delta_n + 48\sum_{i=1}^n \frac{\lg\Delta_i}{\Delta_i}.
\end{equation}
\end{proposition}

\begin{proof}
Let $T_i$ denote the number of sampled bits in iteration $i$.
Suppose we condition on all random variables
$X_1, \dots, X_{i-1}, M_1, \dots, M_{i-1}$, as well as any external
randomness determining the
distributions $F_1, \dots, F_i$.
Then $Z_{i-1} \sim \Uniform \set{0,\dots, M_{i-1}-1} $.
By \cref{thm:entropy-bound-conditional-surprisal} there holds
\begin{equation}
\expect{T_i - H(F_i) + \lg(M_{i-1}/M_i)} \leq \frac{48 \lg \Delta_i}{\Delta_i}.
\end{equation}

Summing over $i = 1, \dots, n$ and using iterated expectations,
the unconditioned expectation satisfies
\begin{align}
L_n + \expect{\lg(M_0/M_n)}
= \expect{\sum_{i=1}^n \left( T_i - H(F_i) + \lg(M_{i-1}/M_i) \right) }
\leq 48 \sum_{i=1}^n \frac{\lg \Delta_i}{\Delta_i}.
\end{align}

Since $M_0 = 1$, we have
$\mathbb E[\lg(M_0/M_n)] = -\mathbb E[\lg M_n]$, so
\begin{equation}
L_n
\le \expect{\lg M_n} + 48\sum_{i=1}^n \frac{\lg\Delta_i}{\Delta_i}
\le \lg\Delta_n + 48\sum_{i=1}^n \frac{\lg\Delta_i}{\Delta_i},
\end{equation}
where the last inequality uses $M_n \le \Delta_n$.
\end{proof}

\getkeytheorem{entropy-bound-overall}

\begin{proof}
Let $m \defeq \ceil{0.1/\epsilon}$ (infinite if $\epsilon = 0$).
Thus the value of $\Delta$ used in iteration $i$ is
\begin{equation}
    \Delta_i = D_{\min\{i,m\}},
    \qquad
  \text{ for }  D_j \defeq 5000j\lceil \lg(j+2)\rceil
    \quad (j\ge 1).
\end{equation}
It is readily verified that
\begin{equation}
\lg D_j \le 10\lg(j+2)
\quad \mbox{ and } \quad
48\thinspace \frac{\lg D_j}{D_j}  \le \frac{0.1}{j}.
\label{eq:mnestic}
\end{equation}

By \cref{prop:entropy-loss}, we have
\begin{equation}
L_n \le
  \lg\Delta_n +
  48\left[
    \sum_{i=1}^{\min \{ n,m \}} \frac{\lg\Delta_i}{\Delta_i} +
    \sum_{i=m+1}^n \frac{\lg\Delta_m}{\Delta_m}
    \right].
\end{equation}
Since $\Delta_n=D_{\min\set{n,m}}$, the first bound in \cref{eq:mnestic} gives
\begin{equation}
\lg\Delta_n \le 10\lg(n+2).
\end{equation}
For the summation, applying the second bound in \cref{eq:mnestic}
to the increasing part of the summation gives
\begin{equation}
    48 \thinspace \sum_{i=1}^{\min\set{n,m}}
    \frac{\lg\Delta_i}{\Delta_i}
    < 0.1\sum_{i=1}^{\min\{n,m\}}\frac1i
    \le 0.1(1+\ln n)
    \le \lg(n+2).
\end{equation}
Finally, if $n>m$, applying the second bound in \cref{eq:mnestic}
with $j = m$ and $0.1/m \le \epsilon$ gives
\begin{equation}
    48\sum_{i=m+1}^n \frac{\lg\Delta_m}{\Delta_m}
    \le (n-m)\thinspace 48 \thinspace \frac{\lg\Delta_m}{\Delta_m}
    \le (n-m)\thinspace \frac{0.1}{m}
    \le (n-m) \epsilon
    \le \epsilon n.
\end{equation}
Combining the three bounds indeed gives
\begin{equation}
    L_n
    \le
    10\lg(n+2)+\lg(n+2)+\epsilon n
    \le
    \epsilon n+11\lg(n+2). \qedhere
\end{equation}
\end{proof}

\begin{corollary}
\label{cor:runtime-bound-overall}

Suppose that the oracle family is uniformly efficient, i.e., that all the oracles $F_i$ use a single upper-bound cost function $\CF[F]{\cdot}{\cdot}$ which
satisfies \cref{eq:oracle-index-efficient,eq:oracle-accuracy-efficient}. If there exists a finite deterministic value $\mu_n$ such that
$\expect{X_i \mid F_i} \le \mu_n$ $(1\le i\le n)$ a.s.,
then the expected runtime over $n$ iterations
of the sampling loop in \cref{alg:loop} is
\begin{equation}
O\!\left( n \thinspace \CF{\mu_n}{\min\set{n, 1/\epsilon}} + (n \log \mu_n) \thinspace \CF{\mu_n}{\mu_n} \right).
\end{equation}
\end{corollary}

\begin{proof}
Sum the bounds from \cref{cor:efficient-runtime-single-sample} over all iterations
$i = 1, \dots, n$, using the fact that $\Delta_i$ is nondecreasing.
Since the oracles are efficient,
we have $\CF{\mu_n}{\Delta_i} \leq O( \CF{\mu_n}{\min\set{i,1/\epsilon}} )$.
\end{proof}

\begin{proposition}
\label{prop:space-bound-overall}
If there exists a finite deterministic value $\mu_n$ such that
$\expect{X_i \mid F_i} \le \mu_n$ $(1\le i\le n)$ a.s.,
then the expected overall space usage over $n$ iterations of the sampling loop in
\cref{alg:loop} is $O(\log n+\log\mu_n)$.
\end{proposition}

\begin{proof}
Let $p_i = f_i(X_i)$ where $f_i$ is the PMF of $F_i$.
Let $S_i$ represent the space usage of iteration $i$ of the sampling loop, and let
$S_{1:n} \defeq \max_{1 \le i \le n} S_i$.
If we condition on $F_i$ and $X_i$, then \cref{prop:sgeomtail} gives
\begin{equation}
\Pr \left( S_i \geq \lambda + c \log(X_i + \Delta_i/p_i) \mid F_i, X_i \right) \leq O( 2^{-\lambda} ).
\end{equation}

By a union bound over the indices $i = 1, \dots, n$, and the fact
that $\log \Delta_i \leq O( \log n)$ for each $i$, we have
\begin{equation}
\expect{S_{1:n}} \leq O\!\left( \log n + \expect{\max_{1 \le i \le n} \log X_i} + \expect{\max_{1 \le i \le n} \log(1/p_i)} \right).
\end{equation}

Now consider an iteration $i$, and suppose we condition on the distribution $F_i$.
By Markov's inequality,
\begin{equation}
\Pr\left( \ln X_i \geq t + \ln \mu_n \mid F_i \right)
= \Pr\left( X_i \geq e^{t} \mu_n \mid F_i \right)
\leq e^{-t}.
\end{equation}
By a union bound, we have $\Pr(\max_i \ln X_i \geq t+\ln \mu_n) \leq n e^{-t}$, and hence again
\begin{equation}
\expect{\max_{1 \le i \le n} \log X_i} \leq O(\log n + \log \mu_n).
\end{equation}

Similarly, as we show in \cref{eq:taillem}, we have
$\Pr \left( \log(1/p_i) \geq t + 3 \log (\mu_n + 3) \mid F_i \right) \leq e^{-t/3}$,
and hence by completely analogous arguments
\begin{equation}
\expect{ \max_{1 \le i \le n} \log(1/p_i) } \leq O(\log n + \log \mu_n). \qedhere
\end{equation}
\end{proof}

\section{Space Lower Bound}
\label{sec:space}

Here, we show that the overall space usage of any
online sampler accessing target distributions via CDF
oracles must grow as $\Omega(\log n + \expect{\max_{1 \le i \le n} \log X_i})$.
By \cref{prop:space-bound-overall}, our algorithm matches this lower
bound provided that $\mu_n \le \poly(n)$, in which case it is optimal up to
constant factors.

To establish the result, we first prove a useful lemma which states that
determining $F(0)$ up to $k$ bits of precision requires oracle calls
$\phi(i,j)$ with $j \geq k - 1$.

\begin{lemma}
\label{lemma:oracle-calls}
Let $F$ be any distribution with a CDF oracle $\phi$,
and let $k\in\Nat$.
There exist rational numbers $a, b$ with
$b-a\ge 2^{-k}$ such that, for each $c\in(a,b)$, there is a
distribution $G$ with $G(0)=c$ and a CDF oracle
$\psi$ agreeing with $\phi$ on all queries with precision parameter
at most $k$.
\end{lemma}

\begin{proof}
Define
\begin{equation}
a \defeq \max\bigl(\set{0} \cup \set{\phi(0,j)-2^{-j}:0\le j\le k}\bigr),
\qquad
b\defeq \min\bigl(\set{1} \cup \set{\phi(i,j)+2^{-j}:i\in\Nat,\;0\le j\le k}\bigr).
\end{equation}
These extrema are well defined: for $0\le j\le k$, the oracle values
$\phi(i,j)$ have denominator $2^j$ and lie in a bounded interval, so
only finitely many values can appear in the displayed sets.
Moreover, all values appearing in the definitions of $a$ and $b$ are
integer multiples of $2^{-k}$.

The oracle guarantees give
$\phi(0,j)-2^{-j}  < F(0) \leq F(i) < \phi(i,j)+2^{-j}$ for $i \geq 0$.
Thus, if $F(0)\in(0,1)$, then
$a<F(0)<b$. If $F(0)=0$, then $a=0<b$, while if $F(0)=1$, then
$a<1=b$. In all cases, $a<b$. Since $a$ and $b$ are integer
multiples of $2^{-k}$, it follows that $b-a\ge 2^{-k}$.

Now fix any $c\in(a,b)$. Define a CDF $G$ by
$G(0)\defeq c$,
and
$G(i)\defeq \max\set{c,F(i)}$ for $i\ge1$.
This function is nondecreasing, satisfies $G(-1)=0$, and has
$\lim_{i\to\infty}G(i)=1$, so it is a valid CDF on $\Nat$.
We define the oracle $\psi$ for $G$ as follows: for $j \leq k$, we set
$\psi(i,j) = \phi(i,j)$ and for $j > k$ we set $\psi(i,j)$ to be a nearest
dyadic rational with denominator $2^j$.
The oracle is clearly valid for $j > k$. So consider a query $\psi(i,j)$ with $j \leq k$.  For
$i=0$, since $c\in(a,b)$, the definitions of $a$ and $b$ give
\begin{align}
    \phi(0,j)-2^{-j}<c<\phi(0,j)+2^{-j}
    && (0\le j\le k),
\end{align}
and hence $|\phi(0,j)-G(0)|<2^{-j}$. For $i\ge1$, if
$G(i)=F(i)$, then validity follows from the oracle guarantee for
$F$. Otherwise $G(i)=c>F(i)$, and for $0\le j\le k$ we have
\begin{align}
    \phi(i,j)-2^{-j}<F(i)<c<b\le \phi(i,j)+2^{-j}.
\end{align}
Thus $\abs{\phi(i,j)-G(i)}<2^{-j}$ in this case as well.
\end{proof}

\begin{lemma}
\label{lemma:bernoulli-oracle-calls}
Let $\mathcal S$ be an exact online random sampling algorithm that
accesses target distributions via CDF oracles.
Let
$F_1,\ldots,F_n$ be arbitrary distributions on $\Nat$, each
equipped with a CDF oracle.
For any $k\in\Nat$, let $E$ denote the event that $\mathcal S$ makes no
oracle queries with precision parameter above $k$ during the first $n$
sampling steps when run on $(F_1,\ldots,F_n)$.
Then $\Pr(E)\le (1-2^{-k})^n$.
\end{lemma}

\begin{proof}
For each $i$, let $a_i,b_i$ be the values from
\cref{lemma:oracle-calls} applied to $F_i$. Let
$X\defeq (X_1,\ldots,X_n)$ be the outputs of $\mathcal S$ when run on
$(F_1,\ldots,F_n)$, and define the associated Bernoulli variables
$\widetilde X_i \defeq  \mathbf 1_{\set{X_i>0}} $.

Fix any binary sequence $z\defeq (z_1,\ldots,z_n)\in\set{0,1}^n$. Let
$\delta\in(0,2^{-k-1})$, and define
\begin{equation}
    c_i \defeq
    \begin{cases}
        a_i+\delta, & z_i=0,\\
        b_i-\delta, & z_i=1.
    \end{cases}
\end{equation}
Here each $c_i\in(a_i,b_i)$. Therefore,
by \cref{lemma:oracle-calls}, there is a distribution
$G_i$ with $G_i(0)=c_i$ whose CDF oracle agrees with that of
$F_i$ on all queries with precision parameter at most $k$.

Now consider coupled executions of $\mathcal S$ on the oracles
$(F_1,\ldots,F_n)$ and $(G_1,\ldots,G_n)$; let $Y\defeq (Y_1, \dots, Y_n)$
be the outputs in the run on $(G_1,\ldots,G_n)$, and set
$\widetilde Y_i \defeq  \mathbf 1_{\set{Y_i>0}}$.
On the event $E$,  the
oracles for $F_i$ and $G_i$ agree on all queries made by $\mathcal S$. So the two
coupled executions have identical transcripts; in particular,
$E \implies X=Y \implies \widetilde X = \widetilde Y$.
Therefore,
\begin{equation}
\Pr\left(E \wedge \left(\widetilde X=z\right)\right)
    \le \Pr\left(\widetilde Y=z\right)
    =
    \prod_{i:z_i=0} c_i
    \prod_{i:z_i=1} (1-c_i)
    =
    \prod_{i:z_i=0} (a_i+\delta)
    \prod_{i:z_i=1} (1-b_i+\delta).
\end{equation}
Since this inequality holds for every $\delta\in(0,2^{-k-1})$, letting
$\delta \rightarrow 0$ gives
\begin{equation}
\Pr\left(E \wedge \left(\widetilde X = z\right)\right)
    \le
    \prod_{i:z_i=0} a_i
    \prod_{i:z_i=1} (1-b_i).
\end{equation}
Summing over $z\in\set{0,1}^n$, we obtain
\begin{align}
\Pr(E)
    =
    \sum_{z\in\set{0,1}^n} \Pr\left(E\wedge\left(\widetilde X = z\right)\right)
    \le
    \sum_{z\in\set{0,1}^n}
    \prod_{i:z_i=0} a_i
    \prod_{i:z_i=1} (1-b_i)
    = \prod_{i=1}^n \bigl(a_i+1-b_i\bigr) \leq \left(1-2^{-k}\right)^n,
\end{align}
where the final inequality uses $b_i-a_i\ge 2^{-k}$ for every $i$.
\end{proof}

We state our main lower bound result:
\begin{theorem}
\label{thm:space-lower-bound}
The expected overall space usage of any online sampling algorithm
accessing the target distributions via CDF oracles
grows over $n$ iterations as
\begin{equation}
\Omega\left(\log n + \expect{\max_{1 \leq i \leq n} \log X_i} \right).
\end{equation}
\end{theorem}
\begin{proof}
For the dependence on $X_i$, note that the (random) bit length of the output
$X_i$ is $\Theta(\log X_i)$, and so the space usage over $n$ iterations
is at least $\Omega(\expect{\max_{1 \leq i \leq n} \log X_i})$.

For the dependence on $n$, note that by \cref{lemma:bernoulli-oracle-calls},
any such sampling algorithm must make an oracle call with precision above $k$ bits
with probability at least $1 - (1-2^{-k})^n$.
In our model, such a query requires at least $k$ bits of space to represent the answer.
So the maximum space usage $S$ over the $n$ iterations satisfies
\begin{equation}
\Pr(S \geq k) \geq 1 - \left(1-2^{-k}\right)^n
\end{equation}
for each $k \in \Nat$, and
\begin{equation}
\expect{S}
\geq \sum_{k=1}^{\infty} \left( 1 - \left(1-2^{-k}\right)^n \right)
\geq \Omega(\log n). \qedhere
\end{equation}
\end{proof}

\appendix

\section{Algorithm ApproxDivision: Proof of \titlezcref{prop:approx-div}}
\label{sec:approx-div}

\begin{algorithm}[H]
\caption{Approximately dividing two arbitrary-precision integers.}
\label{alg:approx-div}
\begin{algorithmic}[1]
\Require Rational $q = a/b \in [0,1]$ and integer $k \geq 1$
\Ensure Rational number $q' = a'/b'$ such that $\abs{q - q'} < 2^{-k}$ and $0 \le a' \le b' < 2^{k+1}$
\Procedure{ApproxDivision}{$q,k$}
\If{$b < 2^{k+1}$}
    \State \Return $a/b$
\Else
    \State Write $b = (b_{n-1} \dots b_0)_2$, with $b_{n-1} = 1$
    \State Write $a = (a_{n-1} \dots a_0)_2$, padded with leading zeros if necessary
    \State $b' \gets (b_{n-1} \dots b_{n-1-k})_2$
    \State $a' \gets (a_{n-1} \dots a_{n-1-k})_2$
    \State \Return $a'/b'$
\EndIf
\EndProcedure
\end{algorithmic}
\end{algorithm}

The case $b < 2^{k+1}$ is clear.
So suppose $b \geq 2^{k+1}$ and let $s \defeq n-(k+1)$.
Observe that $b' = \floor{b/2^s} \ge 2^k$ and $a' = \floor{a/2^s}$.
Now write
\begin{align}
a = 2^s(a' + \alpha), \quad b = 2^s(b' + \beta),
\end{align}
for some $0 \le \alpha, \beta < 1$.
The bound $\abs{a/b-a'/b'} < 2^{-k}$ follows immediately:
\begin{equation}
\abs*{\frac{a}{b} - \frac{a'}{b'}}
= \abs*{\frac{a' + \alpha}{b' + \beta} - \frac{a'}{b'}}
= \frac{\abs*{b'\alpha - a'\beta}}{b'(b'+\beta)}
< \frac{b'}{b'(b'+\beta)}
< \frac{1}{b'+\beta}
< \frac{1}{b'}
\le \frac{1}{2^k},
\end{equation}
where the first inequality follows from the fact
that $0 \le a' \le b'$ and $0 \le \alpha, \beta < 1$,
so that $\abs*{b'\alpha - a'\beta} < b'$.

By our assumption, the leading $k+1$ bits of $b$ and $a$ lie in a
consecutive block of at most $\ceil{(k+1)/w}+1$ limbs, so the algorithm's
computational cost is $O(\ceil{(k+1)/w}+1) = O(k)$.

\section{Algorithm EstimatePMF: Proof of \titlezcref{prop:estimate-pdf}}
\label{sec:pmf}

\begin{algorithm}[H]
\caption{Estimating the PMF $f$ given an oracle for a CDF $F$.}
\label{alg:estimate-pdf}
\begin{algorithmic}[1]
\Require Oracle for a CDF $F$; integer $x \in \Nat$; integer $\nu \geq 1$
\Ensure Estimate $\hat p$ of the PMF
    $f(x)$ such that
    $\hat p\le f(x) \le (1+1/\nu) \hat p$
\Procedure{EstimatePMF}{$F, x, \nu$}
\For{$i = 0, 1, 2, \dots$}                       \label{algline:estimate-pdf-eps}
\State $t \gets \Fl{x}{i} - \Fu{x-1}{i}$    \label{algline:estimate-pdf-fhat}
\If{$2^{-i+1} \nu \leq t$}
    \State \Return $t$                                \label{algline:estimate-pdf-retval}
\EndIf
\EndFor
\EndProcedure
\end{algorithmic}
\end{algorithm}

Let $\hat p = t$ be the return value (\cref{algline:estimate-pdf-retval})
at the final iteration $i$.
We first show that $\hat p \le p \le (1+1/\nu) \hat p$.
From the oracle bound, each iteration $i$ satisfies
\begin{align}
p = F(x) - F(x-1) \ge \Fl{x}{i} - \Fu{x-1}{i} = t = \hat p,
\end{align}
which confirms the lower bound. Also, at the final iteration we have
\begin{align}
p &
\le \Fu{x}{i} - \Fl{x-1}{i}
= t + 2^{-i+1} \label{eq:ti-eq} \le (1 + 1/\nu) t,
\end{align}
where the last inequality follows from the halting condition $t \geq 2^{-i+1} \nu$.

We next establish the runtime.
At iteration
$i = i_{*} \defeq  \ceil{2 + \lg(\nu/p)}$, we have
\begin{align}
t \cdot 2^{i}
&=  2^{i} \left( \Fl{x}{i} - \Fu{x-1}{i} \right)
  \geq 2^{i} (p - 2^{-i+1})= 2^{i} p - 2
  \ge 4 \nu - 2
  \geq 2 \nu
\end{align}
and so the algorithm terminates at this iteration or earlier.
In each iteration, the oracle is invoked with inputs
$(x, i+1)$ and $(x-1, i+1)$, where $i \leq i_* \leq \lg(8 \nu/p)$.
These use oracle inputs $k \leq \lceil \lg(8 \nu/p) \rceil + 1$, as
required by the definition of $\CF{x}{8 \nu/p}$.
The arithmetic costs are also less than a constant multiple of  $\mathcal C_F$.

\section{Technical Lemmas}
\label{sec:lemmas}

\begin{lemma}
\label{lemma:ff}
Let $h: \Real_{\geq 0}^2 \rightarrow \Real_{\geq 0}$
be any function which is nondecreasing
in both arguments and for which there exist values $C, k \geq 1$ with
\begin{align}
\label{eq:ff33}
\frac{h(x,y)}{h(x',y')} \leq C \left( \frac{\log x \log y}{\log x' \log y'} \right)^k
&& (0 \leq x' \leq x, 0 \le y' \leq y)
\end{align}

Let $a_0, a_1, \dots$ be nonnegative real numbers with
$\sum_i a_i = 1$. Then
\begin{equation}
\sum_{i=0}^{\infty} a_i h ( i, 1/a_i )
\leq c  \sum_{i=0}^{\infty} a_i h( i, i)
\end{equation}
for some factor $c$ that depends on $k$ and $C$.
\end{lemma}
\begin{proof}
Throughout this proof, we regard $C$ and $k$ as constants; all asymptotic
notations may depend on these parameters.
Thus,  applying \cref{eq:ff33}
with $x' = y' = 0$ and any values $x, y \geq 0$ gives
\begin{equation}
\label{eq:ff34}
h(x,y) \leq O\!\left(h(0,0)\log^k x \log^k y\right).
\end{equation}

Partition $\Nat = A_1 \cup A_2 \cup A_3$, where
\begin{align}
A_1 = \set*{ i : i \leq e^{k} }, \qquad A_2  = \set*{i: i > e^{k}, a_i \geq (i+2)^{-2} }, \qquad A_3 = \set*{ i: i > e^{k}, a_i < (i+2)^{-2} }.
\end{align}

Clearly, $\abs{A_1} \leq \ceil{e^{k}} = O(1)$ and $i \leq e^k \leq O(1)$ for $i \in A_1$.
Using \cref{eq:ff34}, the contribution of $A_1$ is
\begin{align}
\label{eq:A1-bound}
\sum_{i \in A_1} a_i h(i, 1/a_i)
\leq O \Bigl( h(0,0)\sum_{i \in A_1} a_i \log^k i \log^{k}(1/a_i) \Bigr)
\leq O \Bigl( h(0,0)\sum_{i \in A_1} a_i \log^{k}(1/a_i) \Bigr).
\end{align}

The function $x \mapsto x \log^k(1/x)$ is bounded above
by $\max\set{1,(k/e)^k}=O(1)$ on $[0,1]$.
Hence, \cref{eq:A1-bound} is at most
\begin{equation}
O(h(0,0)) \sum_{i \in A_1}1 = O( h(0,0)) |A_1|
=O\!\left(h(0,0)\right).
\end{equation}

Next consider the contribution of $A_2$. Here, since $h$ is nondecreasing
in its second argument, we calculate
\begin{align}
\label{eq:nom1}
\sum_{i \in A_2} a_i h(i, 1/a_i)
\leq \sum_{i \in A_2} a_i h\left(i, (i+2)^2\right)
\end{align}

Applying \cref{eq:ff33} with $x = x' = i$, $y = (i+2)^2$, and $y' = i-2$, the
right-hand side of \cref{eq:nom1} is at most
\begin{equation}
O \Bigl( \sum_{i \in A_2} a_i h(i, i) \Bigr)
\leq O \Bigl( \sum_{i=0}^{\infty} a_i h(i, i) \Bigr).
\end{equation}
Note here that $i - 2> e^k - 2 > 0$ since $i \in A_2$.

Finally, for the contribution to $A_3$, we calculate
\begin{align}
\sum_{i \in A_3} a_i h(i, 1/a_i)
\leq O \Bigl(
    h(0,0)\sum_{i \in A_3}
    a_i \log^k i\log^{k}(1/a_i)
    \Bigr).
\end{align}

Consider again the function $x \mapsto  x \log^k(1/x)$. For $x < e^{-k}$, it is increasing in $x$. Here, $a_i < e^{-k}$ and $a_i < (i+2)^{-2} \leq (e^k + 2)^{-2}$ by definition of $A_3$, so
\begin{equation}
a_i \log^k(1/a_i) \leq (i+2)^{-2} \log^k\left((i+2)^2\right),
\end{equation}
and hence
\begin{align}
\sum_{i \in A_3} a_i h(i, 1/a_i)
\leq O \Bigl(
    h(0,0)\sum_{i \in A_3}
    (i+2)^{-2}\log^{2k}(i+2)
    \Bigr)
=O\!\left(h(0,0)\right).
\end{align}

Because $h$ is nondecreasing and $\sum_i a_i=1$, we have
$h(0,0) \leq \sum_{i=0}^{\infty}a_i h(i,i)$. So putting the three cases together, we have
\begin{equation}
\sum_{i=0}^{\infty} a_i h(i, 1/a_i)
\leq
O \Bigl(
    \sum_{i=0}^{\infty}a_i h(i,i)
    \Bigr).
\qedhere
\end{equation}
\end{proof}

\begin{corollary}
\label{lemma:entropy-leq-ElogX}
For any random variable $X$ over $\Nat$ with PMF $f$, the Shannon entropy satisfies
\begin{equation}
H(X) \defeq \expect{\lg(1/f(X))} \leq O(\expect{\log X}).
\qedhere
\end{equation}
\end{corollary}

\begin{proof}
Apply \cref{lemma:ff} with function $h(x,y) = \log y$ and $k = 1$;
here $\expect{\log(1/f(X))} = \sum_i a_i h(i, 1/a_i)$ for $a_i = f(i)$.
So $\expect{\log(1/f(X))} \leq O( \expect{ h(X, X) }) = O(\expect{\log X})$.
\end{proof}

\begin{lemma}
\label{eq:taillem}
For any random variable $X$ supported over $\Nat$ with PMF $f$ and mean $\mu = \expect{X}$, there holds
\begin{equation}
\Pr \left(  \ln(1/f(X)) \geq t + 3 \ln (\mu + 3) \right) \leq e^{-t/3}.
\end{equation}
\end{lemma}
\begin{proof}
Let $Y \defeq 1/f(X)$.
Exponentiating and applying Markov's inequality gives
\begin{equation}
\Pr \bigl( \ln Y \geq t + 3 \ln (\mu + 3) \bigr)
= \Pr \bigl( Y^{1/3} \geq e^{t/3} (\mu + 3) \bigr)
\leq e^{-t/3} (\mu+3)^{-1} \expect{Y^{1/3}}.
\end{equation}

Let $a_i \defeq f(i)$. We bound the expectation of $Y^{1/3}$ as
\begin{align}
\expect{Y^{1/3}} &= \sum_{i=0}^{\infty} a_i \cdot a_i^{-1/3}
= \sum_{i: a_i < (i+1)^{-3}} a_i \cdot a_i^{-1/3} + \sum_{i: a_i \geq (i+1)^{-3}} a_i \cdot a_i^{-1/3} \\
&\leq \sum_{i=0}^{\infty} (i+1)^{-2} + \sum_{i=1}^{\infty} a_i \cdot (i+1) \\
&\leq \pi^2/6 + 1 + \mu.
\end{align}

Plugging in the estimate gives
\begin{equation}
\Pr\left( \ln Y \geq t + 3 \ln (\mu + 3) \right)
\leq e^{-t/3} (\mu+3)^{-1} (\pi^2/6 + 1 + \mu)
\leq e^{-t/3}. \qedhere
\end{equation}
\end{proof}

\section*{Acknowledgments}
This material is based upon work supported by the National Science Foundation
under Award No.~2311983.
Any opinions, findings and conclusions or recommendations expressed in this
material are those of the authors and do not necessarily reflect the views
of the National Science Foundation.

\printbibliography

\end{document}